\documentclass[a4paper,10pt]{article}
\usepackage{amsmath, amssymb}
\usepackage{graphicx}
\usepackage{color}
\usepackage{graphicx}
\usepackage{url}
\usepackage{color,soul}
\usepackage{float}
\usepackage{amsmath}
\usepackage{amssymb}
\usepackage{courier}
\usepackage{multicol}
\usepackage[abs]{overpic}
\usepackage{enumitem}
\usepackage{multicol}
\usepackage{amsthm}
\usepackage{changepage}
\usepackage{framed}
\usepackage{mdframed}
\usepackage{lipsum}
\usepackage{longtable}
\usepackage{tablefootnote}
\usepackage{threeparttable}
\newmdenv[
linecolor=gray,
linewidth=3pt,
topline=false,
bottomline=false,
rightline=false,
rightmargin=-10pt,
leftmargin=0pt,
skipabove=\topsep,
skipbelow=\topsep
]{siderules}

\newtheorem{theorem}{Theorem}
\newtheorem{lemma}[theorem]{Lemma}
\newtheorem{claim}[theorem]{Claim}

\usepackage[linesnumbered,boxed]{algorithm2e}
\allowdisplaybreaks
\usepackage{pgf,tikz}
\usepackage{mathrsfs}
\usetikzlibrary{arrows}
\usepackage{times}
\DeclareMathSizes{10}{9}{6}{5}

\def\Tr{\text{Tr}}

\def\RR{\mathbb R}

\def\ZZ{\mathbb Z}

\def\SS{\mathbb S}

\def\cA{\mathcal A}

\def\cD{\mathcal D}

\def\cK{\mathcal K}
\def\cS{\mathcal S}

\def\cT{\mathcal T}
\def\cQ{\mathcal Q}

\def\cL{\mathcal L}

\def\cX{\mathcal X}

\def\b1{\mathbf 1}

\def\eps{\varepsilon}
\def\G {G=(V,E)}

\def\bareps{{\eps_s}}
\def \P {\textsc{Packing}}
\def \C {\textsc{Covering}}

\def \nP {\textsc{Norm-Packing}}
\def \nC {\textsc{Norm-Covering}}
\def \rP {\textsc{Rbst-Packing}}
\def \rC {\textsc{Rbst-Covering}}

\def \mnmx {\textsc{MinMax}}
\def \mxmn {\textsc{MaxMin}}

\newcommand{\poly}{\operatorname{poly}}
\newcommand{\polylog}{\operatorname{polylog}}

\newcommand{\argmin}{\operatorname{argmin}}
\newcommand{\argmax}{\operatorname{argmax}}
\newcommand{\diag}{\operatorname{diag}}
\newcommand{\supp}{\operatorname{supp}}

\newcommand{\im}{\operatorname{image}}
\newcommand{\nll}{\operatorname{null}}

\newtheorem{fact}{Fact}

\newtheorem{remark}{Remark}
 
\newtheorem{definition}{Definition}

\newcommand{\raf}[1]{(\ref{#1})}

\usepackage[top=2.11cm, bottom=2.3cm, left=2cm, right=2cm]{geometry}

\title{Finding Sparse Solutions for Packing and Covering Semidefinite Programs} 
\author{
	Khaled Elbassioni\thanks{Masdar Institute, Khalifa University of Science and Technology, P.O. Box 54224, Abu Dhabi, UAE;
		(kelbassioni@masdar.ac.ae)}
	\and
	Kazuhisa Makino\thanks{Research Institute for Mathematical Sciences (RIMS)
		Kyoto University, Kyoto 606-8502, Japan;
		(makino@kurims.kyoto-u.ac.jp)}
}
\date{}
\begin{document}
\maketitle

\begin{abstract}
Packing and covering semidefinite programs (SDPs) appear in natural relaxations of many combinatorial optimization problems as well as a number of other applications. Recently, several techniques were proposed, that utilize the particular structure of this class of problems, to obtain more efficient algorithms than those offered by general SDP solvers. For certain applications, such as those described in this paper, it may be desirable to obtain {\it sparse} dual solutions, i.e., those with support size (almost) independent of the number of primal constraints. In this paper, we give an algorithm that finds such solutions, which is an extension of a {\it logarithmic-potential} based algorithm of  Grigoriadis, Khachiyan, Porkolab and Villavicencio (SIAM Journal of Optimization 41 (2001)) for packing/covering linear programs.
\end{abstract}
\section{Introduction}\label{Intro}

\subsection{Packing and Covering SDPs}

We denote by $\SS^n$ the set of all $n\times n$ real symmetric matrices and by  $\SS^n_+\subseteq\SS^n$ the set of all $n\times n$ positive semidefinite matrices.
Consider the following pairs of {\it packing-covering} semidefinite programs (SDPs):

{\centering \hspace*{-18pt}
	\begin{minipage}[t]{.47\textwidth}
		\begin{alignat}{3}
		\label{packI}
		\tag{\P-I} \quad& \displaystyle z_I^* = \max\quad C\bullet X\\
		\text{s.t.}\quad & \displaystyle A_i\bullet X\leq b_i, \forall i\in [m]\nonumber\\
		\qquad &X\in\RR^{n\times n},~X\succeq  0\nonumber
		\end{alignat}
	\end{minipage}
	\,\,\, \rule[-14ex]{1pt}{14ex}
	\begin{minipage}[t]{0.47\textwidth}
		\begin{alignat}{3}
		\label{coverI}
		\tag{\C-I} \quad& \displaystyle z_I^* = \min\quad b^Ty\\
		\text{s.t.}\quad & \displaystyle \sum_{i=1}^my_iA_i\succeq C\nonumber\\
		\qquad &y\in\RR^m,~y\geq 0\nonumber
		\end{alignat}
\end{minipage}}

{\centering \hspace*{-18pt}
	\begin{minipage}[t]{.47\textwidth}
		\begin{alignat}{3}
		\label{coverII}
		\tag{\C-II} \quad& \displaystyle z_{II}^* = \min\quad C\bullet X\\
		\text{s.t.}\quad & \displaystyle A_i\bullet X\geq b_i, \forall i\in [m]\nonumber\\
		\qquad &X\in\RR^{n\times n},~X\succeq  0\nonumber
		\end{alignat}
	\end{minipage}
	\,\,\, \rule[-14ex]{1pt}{14ex}
	\begin{minipage}[t]{0.47\textwidth}
		\begin{alignat}{3}
		\label{packII}
		\tag{\P-II} \quad& \displaystyle z_{II}^* = \max\quad b^Ty\\
		\text{s.t.}\quad & \displaystyle \sum_{i=1}^my_iA_i\preceq C\nonumber\\
		\qquad &y\in\RR^m,~y\geq 0\nonumber
		\end{alignat}
\end{minipage}}

\noindent where $C,A_1,\ldots,A_m \in \SS_+^n$ are (non-zero) positive semidefinite matrices, and $b=(b_1,\ldots,b_n)^T\in\RR^m_+$ is a nonnegative vector.  In the above, $C\bullet X:=\Tr(CX)=\sum_{i=1}^n\sum_{j=1}^n c_{ij}x_{ij}$, and
"$\succeq$`` is the {\it L\"owner order} on matrices:  $A\succeq B$ if and only if $A-B$ is positive semidefinite. This type of SDPs arise in many applications, see, e.g. \cite{IPS11,IPS10} and the references therein. 

We will make the following assumption throughout the paper: 
\begin{itemize}
	\item[{\bf (A)}] $b_i>0$ and hence $b_i=1$ for all $i\in[m]$.
\end{itemize} 
It is known that, under assumption (A), {\it strong duality} holds for  problems \raf{packI}-\raf{coverI} (resp., \raf{packII}-\raf{coverII}) (see Appendix~\ref{normal} for details).

Let $\epsilon\in(0,1]$ be a given constant. We say that $(X,y)$ is an \textit{$\epsilon$-optimal} primal-dual solution for \raf{packI}-\raf{coverI} if $(X,y)$ is a primal-dual feasible pair such that 
\begin{align}\label{pdf-I}
C\bullet X\geq(1-\epsilon)b^Ty\ge (1-\epsilon)z^*_I.
\end{align}
Similarly, we  say that $(X,y)$ is an $\epsilon$-optimal primal-dual solution for \raf{packII}-\raf{coverII} if $(X,y)$ is a primal-dual feasible pair such that 
\begin{align}\label{pdf-II}
C\bullet X\leq(1+\epsilon)b^Ty\le (1+\epsilon)z^*_{II}.
\end{align}
  
Since in this paper we allow the number of constraints $m$ in \raf{packI} (resp.,  \raf{coverII}) to be {\it exponentially} (or even infinitely) large, we will assume the availability of the following {\it oracle}: 

\begin{description}
	\item[Max$(Y)$ (resp., Min$(Y)$)]: Given $Y\in\SS_+^n$, find $i\in\argmax_{i\in[m]}A_i\bullet Y$ (resp., $i\in\argmin_{i\in[m]}A_i\bullet Y$).
\end{description}	
Note that an {\it approximation} oracle computing the maximum (resp., minimum) above within a factor of $(1-\epsilon)$ (resp., $(1+\epsilon)$) is also sufficient for our purposes. 

A primal-dual solution $(X,y)$ to \raf{coverI} (resp., \raf{packII})  is said to be $\eta$-{\it sparse}, if the size of $\supp(y):=\{i\in[m]:y_i>0\}$ is at most $\eta$. Our objective in this paper is to develop {\it primal-dual} algorithms that find sparse $\epsilon$-optimal solutions for~\raf{packI}-\raf{coverI} and~\raf{packII}-\raf{coverII}.

\subsection{Reduction to Normalized Form}\label{normal-}
When $C=I=I_n$, the identity matrix in $\RR^{n\times n}$ and $b=\b1$, the vector of all ones in $\RR^m$, we say that the packing-covering SDPs are in {\it normalized} form:

{\centering \hspace*{-18pt}
	\begin{minipage}[t]{.47\textwidth}
		\begin{alignat}{3}
		\label{npackI}
		\tag{\nP-I} \quad& \displaystyle z_I^* = \max\quad I\bullet X\\
		\text{s.t.}\quad & \displaystyle A_i\bullet X\leq 1, \forall i\in [m]\nonumber\\
		\qquad &X\in\RR^{n\times n},~X\succeq  0\nonumber
		\end{alignat}
	\end{minipage}
	\,\,\, \rule[-14ex]{1pt}{14ex}
	\begin{minipage}[t]{0.47\textwidth}
		\begin{alignat}{3}
		\label{ncoverI}
		\tag{\nC-I} \quad& \displaystyle z_I^* = \min\quad \b1^Ty\\
		\text{s.t.}\quad & \displaystyle \sum_{i=1}^my_iA_i\succeq I\nonumber\\
		\qquad &y\in\RR^m,~y\geq 0.\nonumber
		\end{alignat}
\end{minipage}}

{\centering \hspace*{-18pt}
	\begin{minipage}[t]{.47\textwidth}
		\begin{alignat}{3}
		\label{ncoverII}
		\tag{\nC-II} \quad& \displaystyle z_{II}^* = \min\quad I\bullet X\\
		\text{s.t.}\quad & \displaystyle A_i\bullet X\geq 1, \forall i\in [m]\nonumber\\
		\qquad &X\in\RR^{n\times n},~X\succeq  0\nonumber
		\end{alignat}
	\end{minipage}
	\,\,\, \rule[-14ex]{1pt}{14ex}
	\begin{minipage}[t]{0.47\textwidth}
		\begin{alignat}{3}
		\label{npackII}
		\tag{\nP-II} \quad& \displaystyle z_{II}^* = \max\quad \b1^Ty\\
		\text{s.t.}\quad & \displaystyle \sum_{i=1}^my_iA_i\preceq I\nonumber\\
		\qquad &y\in\RR^m,~y\geq 0.\nonumber
		\end{alignat}
\end{minipage}}

\medskip

\noindent In the appendix, we show that, at the loss of a factor of $(1+\epsilon)$ in the objective, any pair of packing-covering SDPs of the form~\raf{packI}-\raf{coverI} can be brought in $O(n^3)$, increasing the oracle time only by $O(n^\omega)$, where $\omega$ is the exponent of matrix multiplication, to the normalized form~\raf{npackI}-\raf{ncoverI}, under the following assumption:

\begin{itemize}
	\item[\bf (B-I)] There exist $r$ matrices, say $A_{1},\ldots,A_{r}$, such that $\bar A:=\sum_{i=1}^rA_{i}\succ 0$. In particular, $\Tr(X)\le \tau:=\frac{r}{\lambda_{\min}(\bar A)}$ for any optimal solution $X$ for \raf{packI}.
\end{itemize}

%The following stronger assumption is made in \cite{}:
%\begin{itemize}
%	\item[\bf (B'-I)] $\Tr(X)\le \tau$ for any optimal solution $X$ for \raf{packI}, where $\tau$ is a given positive number. Hence, w.l.o.g., $A_1=\frac{1}{\tau}I$ in \raf{npackI}.
%\end{itemize}
%To see that assumption (B-I) implies (B'-I), let $X$ be an optimal solution to \raf{packI}. Then $
%I\bullet X\le \frac{\bar A\bullet X}{\lambda_{\min}(\bar A)}\le\tau:=\frac{r}{\lambda_{\min}(\bar A)}.
%$

Similarly, we show in the appendix (some of the results are reproduced with simplifications from \cite{JY11}) that, at the loss of a factor of $(1+\epsilon)$ in the objective, any pair of packing-covering SDPs of the form~\raf{packII}-\raf{coverII} can be brought in $O(n^3)$ time, increasing the oracle time only by $O(n^\omega)$, to the normalized form~\raf{npackII}-\raf{ncoverII}. Moreover, we may assume in this normalized form that 
\begin{itemize}
	\item[({\bf B-II})] $\lambda_{\min}(A_i)=\Omega\big
	(\frac{\epsilon}{n}\cdot\min_{i'}\lambda_{\max}(A_{i'})\big)$ for all $i\in[m]$,
\end{itemize} 
where, for a positive semidefinite matrix $B\in\SS_+^{n\times n}$, we denote by $\{\lambda_j(B):~j=1,\ldots,n\}$ the eigenvalues of  $B$, and by $\lambda_{\min}(B)$ and $\lambda_{\max}(B)$ the minimum and maximum eigenvalues of $B$, respectively. 
With an additional $O(mn^2)$ time, we may also assume that:
\begin{itemize}
	\item[({\bf B-II$'$})] $\frac{\lambda_{\max}(A_i)}{\lambda_{\min}(A_i)}= O\big
	(\frac{n^2}{\epsilon^2}\big)$ for all $i\in[m]$.
\end{itemize}

Thus, from now on we focus on the normalized problems.

\subsection{Main Result and Related Work}

Problems~\raf{packI}-\raf{coverI} and \raf{packII}-\raf{coverII} can be solved using general SDP solvers, such as interior-point methods: for instance, the barrier method (see, e.g., \cite{NN94}) can compute a solution, within an {\it additive} error of $\epsilon$ from the optimal, in time $O(\sqrt{n}m(n^3+mn^2+m^2)\log\frac{1}{\epsilon})$ (see also \cite{A95,VB96}). 
However, due to the special nature of  \raf{packI}-\raf{coverI} and \raf{packII}-\raf{coverII}, better algorithms can be obtained. Most of the improvements are obtained by using {\it first order methods} \cite{AHK05,AK07,AK16,ALO16,GH16,IPS11,JY11,JY12,KL96,N07,PT12,PTZ16}, or second order methods \cite{IPS05,IPS10}. In general, we can classify these algorithms according to whether they are:  
\begin{itemize}
\item[(I)] {\it width-independent}: the running time of the algorithm depends {\it polynomially} on the bit length of the input; for example, in the of case of ~\raf{packI}-\raf{coverI}, the running time is $\poly(n,m,\cL,\log \tau,\frac{1}{\epsilon})$, where $\cL$ is the maximum bit length needed to represent any number in the input; on the other hand, the running time of a width-dependent algorithm will depend polynomially on a``width parameter'' $\rho$, which is polynomial in $\cL$ and $\tau$;

\item[(II)] {\it parallel}: the algorithm takes $\polylog(n,m,\cL,\log \tau)\cdot\poly(\frac{1}{\epsilon})$ time, on a $\poly(n,m,\cL,\log \tau,\frac{1}{\epsilon})$ number of processors;

\item[(III)] {\it output sparse solutions}: the algorithm outputs an $\eta$-sparse solution to \raf{coverI} (resp., \raf{packII}), for $\eta=\poly(n,\log m,\cL,\log\tau,\frac{1}{\epsilon})$ (resp., $\eta=\poly(n,\log m,\cL,\frac{1}{\epsilon})$), where $\tau$ is a parameter that bounds the trace of any optimal solution $X$ (see Section~\ref{normal-} for details);

\item[(IV)] {\it oracle-based}: the only access of the algorithm to the matrices $A_1,\ldots,A_m$ is via the maximization/minimization oracle, and hence the running time is independent of $m$.
\end{itemize}
Table~\ref{t1} below gives a summary\footnote{We provide rough estimates of the bounds, as some of them are not stated explicitly in the corresponding paper in terms of the parameters we consider here.} of the most relevant results together with their classifications, according to the four criteria described above. We note that almost all these algorithms for packing/covering SDP's are generalizations of similar algorithms for packing/covering linear programs (LPs), and most of them are essentially based on an {\it exponential potential function} in the form of {\it scalar exponentiation}, e.g., \cite{AHK05,KL96}, or {\it matrix exponentiation} \cite{AK07,AK16,ALO16,JY12,IPS11}. For instance, several of these results use the scalar or matrix versions of the {\it multiplicative weights updates} (MWU) method (see, e.g., \cite{AHK06}), which are extensions of similar methods for packing/covering LPs \cite{GK07,GK95,Y01,PST91}.

In \cite{GKPV01}, a different type of algorithm was given for covering LPs (indeed, more generally, for a class of concave covering inequalities) based on a {\it logarithmic} potential function. In this paper, we show that this approach can be extended to provide sparse solutions for both versions of packing and covering SDPs.

\begin{table}
	\centering
	\begin{threeparttable}
		\label{t1}
		\caption{Different Algorithms for Packing/covering SDPs}
		{\scriptsize	\begin{tabular}{|l|l|l|l|l|c|c|c|c|}
				\hline
				Paper & Problem & Technique &Most Expensive & \# Iterations & Width- &Parallel &Sparse &Oracle-\\
				&  &  & Operation &  & indep. & & &based\\\hline\hline
				\cite{AHK05,KL96}& \raf{packI} & MWU & $\max/\min$ eigenvalue& $O(\frac{\rho\log m}{\epsilon^2})$ & No & No & No\tnote{$*$}& No
				\\
				& \raf{coverII}&  &  of a PSD matrix $\tilde O(\frac{n^2}{\epsilon})$&  &  &  & &\\ \hline
				\cite{AK16}& \raf{packI} & Matrix MWU & Matrix exponentiation  & $O(\frac{\rho^2\tau^2\log n}{\epsilon^2 (z^*_I)^2})$ &  No & No & No\tnote{$*$}& Yes\\
				& \raf{coverII}&  &  $O(n^3)$&  &  &  & &\\ \hline
				\cite{IPS05,IPS11}& \raf{packI} & Nesterov's smoothing & Matrix exponentiation  & $O(\frac{\tau\log m}{\epsilon})$ & No & No & No& No\\
				& &   technique \cite{N05,N07}&  $O(n^3)$&  & & &  & \\ \hline
				\cite{IPS10}& \raf{coverII} & Nesterov's smoothing & $\min$ eigenvalue   of & $ O(\frac{\rho^2\log (nm)}{\epsilon})$ & No & No & No& No\\
				& &   technique \cite{N05,N07}& a non PSD matrix $O(n^3)$&&  &  &  & \\ \hline
				\cite{JY11}& \raf{packI}\&& MWU  & eigenvalue & $ O(\frac{\log^{13}n\log m}{\epsilon^{13}})$ & Yes & Yes & No &No\\
				& \raf{coverII}  &   technique \cite{N05,N07}& decomposition $O(n^3)$&  &  &  & &\\ \hline
				\cite{PT12,PTZ16}& \raf{packII}\&& Matrix MWU & Matrix exponentiation  & $O(\frac{\log^{3} m}{\epsilon^{3}})$ & Yes & Yes & No & No\\
				& \raf{coverII}&  &  $O(n^3)$&  &  &  & & \\ \hline
				\cite{ALO16}& \raf{packI}\&& Gradient Descent $+$ & Matrix exponentiation  & $O(\frac{\log^{2} (mn)\log\frac1\epsilon}{\epsilon^{2}})$ &  Yes & Yes & No & No\\
				& \raf{coverII}&  Mirror Descent  & $O(n^3)$ & & &  & &\\ \hline
				This paper& \raf{packII}\&& Matrix MWU & Matrix exponentiation  & $O(\frac{n\log n}{\epsilon^{2}})$ & Yes & No & Yes &Yes\\
				Appendix~\ref{MMWU-covering}& \raf{coverII}&  &  $O(n^3)$&  &  &  & & \\ \hline
				This paper & \raf{packII} \&&Logarithmic & Matrix inversion & $O(n\log (n\cL\tau)+\frac{n}{\epsilon^2})$ & Yes & No &Yes& Yes\\
				& \raf{coverII}&potential \cite{GKPV01} & $O(n^{\omega})$ & & & & &\\
				&\raf{packII} \&& & &$O(n\log (n/\epsilon)+\frac{n}{\epsilon^2})$&&&&\\
				&  \raf{coverII}& & & &&&&\\ \hline
		\end{tabular}}
		\begin{tablenotes}
			\item[$*$] In fact, these algorithms find sparse solutions, in the sense that the dependence of the size of the support of the dual solution on $m$ is at most logarithmic; however,  the dependence of the size of the support on the bit length $\cL$ is not polynomial. 
		\end{tablenotes}
	\end{threeparttable}
\end{table}

As we can see from the table, among all the algorithms, the logarithmic-potential algorithm, presented in this paper, is the only one that produces sparse solutions, in the sense described above. We also show in Appendix~\ref{MMWU-covering} that a modified version of the matrix exponential MWU algorithm \cite{AK07} can yield sparse solutions for ~\raf{packII}-\raf{coverII}. However, the overall running time of this matrix MWU algorithm is larger by a factor of (roughly) $\Omega(n^{3-\omega})$ than that of the logarithmic-potential algorithm, where $\omega$ is the exponent of matrix multiplication. Moreover, we were not able to extend the matrix MWU algorithm to solve \raf{packI}-\raf{coverI} (in particular, it seems tricky to bound the number of iterations). 

A work that is also related to ours is the sparsification of graph Laplacians \cite{BSS14} and positive semidefinite sums \cite{SHS16}. Given matrices $A_1,\ldots,A_m\in\SS_+^n$ and $\epsilon>0$, it was shown in \cite{SHS16} that one can find, in $O\big(\frac{n}{\epsilon^2}(n^{\omega}+\cT)\big)$ time, a vector $y\in\RR^m_+$ with support size $O(\frac{n}\epsilon^2)$, such that $B\preceq\sum_iy_iA_i\preceq (1+\epsilon)B$, where $B:=\sum_iA_i$ and $\cT$ is the time taken by a single call to the minimization oracle Min$(Y)$ (for a not necessarily positive semidefinite matrix $Y$). An immediate corollary is that, given an $\epsilon$-optimal solution $y$ for \raf{coverI} (resp., \raf{packI}), one can find in $O\big(\frac{n}{\epsilon^2}(n^{\omega}+\cT)\big)$ time an $O(\epsilon)$-optimal solution $y'$ with support size  $O(\frac{n}{\epsilon^2})$. Interestingly, the algorithm in \cite{SHS16} (which is an extension for the rank-one version in \cite{BSS14}) uses the {\it barrier potential function} $\Phi'(x,F):=\Tr\big((H-xI)^{-1}\big)$ (resp., $\Phi'(x,H):=\Tr\big((xI-H)^{-1}\big)$), while in our algorithms (generalizing the potential function in \cite{GKPV01}) we use the logarithmic potential function $\Phi(x,H)=\ln x+\frac{\epsilon}{n}\ln\det\big(H-x I\big)=\ln x-\frac{\epsilon}{n}\int_x\Phi'(x,H)dx$  (resp., $\Phi(x,H)=\ln x-\frac{\epsilon}{n}\ln\det\big(x I-H\big)=\ln x-\frac{\epsilon}{n}\int_x\Phi'(x,H)dx$). %, which is very close to the {\it integral} of $\Phi'(x,H)$ with respect to $x$.  
Sparsification algorithms with better running times were recently obtained in \cite{ALO15,LS17}. Since the sparse solutions produced by our algorithms may have support size slightly more (by polylogarithmic factors) than $O(\frac{n}{\epsilon^2})$, we may use, in a post-processing step, the sparsfication algorithms, mentioned above, to convert our  solutions to ones with  support size $O(\frac{n^2}\epsilon)$, without increasing the overall asymptotic running time.

To motivate our algorithms, in Section~\ref{app}, we give two applications, mainly in robust optimization, that require finding sparse solutions for a packing/covering SDP.   

\section{A Logarithmic Potential Algorithm}\label{log-alg}

\subsection{Algorithm for~\raf{packI}-\raf{coverI}}\label{sec:log-pack-alg}

In this section we give an algorithm for finding a sparse $O(\epsilon)$-optimal primal-dual solution for~\raf{packI}-\raf{coverI}. 
\paragraph{High-level Idea of the Algorithm.}~
The idea of the algorithm is quite intuitive. It can be easily seen that problem~\raf{ncoverI} is equivalent to finding a convex combination of the $A_i$'s that maximizes the minimum eigenvalue, that is,  $\max_{y\in\RR_+^m:\b1^Ty=1}\lambda_{\min}(F(y))$, where $F(y):=\sum_{i=1}^m y_iA_i$, and $\b1$ is the $m$-dimensional vector of all ones. Since $\lambda_{\min}(F(y))$ is not a {\it smooth} function in $y$, it is more convenient to work with a smooth approximation of it, which is obtained by maximizing (over $x$) a {\it logarithmic potential function} $\Phi(x,F(y))$ that captures the constraints that each eigenvalue of $F(y)$ is at least $x$. The unique maximizer $x=\theta^*$ of $\Phi(x,F(y))$ defines a set  of ``weights" (these are the eigenvalues of the primal matrix $X$ computed in line~\ref{s3.-mnmxI} of the algorithm) such that the weighted average of the $\lambda_j(F(y))$'s is a very close approximation of $\lambda_{\min}(F(y))$. Thus, to maximize  this average (which is exactly $X\bullet F(y)$), we obtain a direction (line~\ref{s3-mnmxI}) along which $y$ is modified with an appropriate step size (line~\ref{s6-mnmxI}).   

For numbers $x\in\RR_+$ and $\delta\in(0,1)$,  a $\delta$-(lower) approximation $x_\delta$ of $x$ is a number such that  $(1-\delta)x\le x_\delta<x$. For $i\in[m]$, $\b1_i$ denotes the $i$th unit vector of dimension $m$. 

The algorithm is shown as Algorithm~\ref{log-pack-alg}. The main while-loop (step~\ref{s1.-mnmxI}) is embedded within a sequence of scaling phases, in which each phase starts from the vector $y(t)$ computed in the previous phase and uses double the accuracy. The algorithm stops when the scaled accuracy $\eps_s$ drops below the desired accuracy $\epsilon\in(0,1/2)$.

%\begin{algorithm}[H]
%	\SetAlgoLined
%	$s \gets 0$; $\epsilon^0\gets 1$; $t^0\gets 0$; $y^0 \gets\frac1m\b1$\label{}\\
%\While{$\epsilon^s>\epsilon$}{
%   $(X^{s+1},y^{s+1},t^{s+1})\gets$Pack$(\epsilon^s,y^s,t^s)$\\
%   $\epsilon^{s+1}\gets\epsilon^s/2$\\
%   $s\gets s+1$
%}
%\caption{Logarithmic-potential Packing Algorithm}\label{log-pack-alg-main}
%\end{algorithm}

\begin{algorithm}[H]
	\SetAlgoLined
	$s \gets 0$; $\eps_0\gets \frac12$; $t\gets 0$; 	 $\nu(0) \gets 1$; $y(0) \gets\frac1r\sum_{i=1}^r\b1_i$\label{s1-mnmxI}\\
	\While{$\bareps>\epsilon$}{\label{s0.-mnmxI}
		$\delta_s\gets\frac{\eps_s^3}{32n}$\\
		\While{$\nu(t) >\bareps$}{ \label{s1.-mnmxI}
			$\theta(t)\gets\theta^*(t)_{\delta_s}$, where $\theta^*(t)$ is the smallest positive number root of the equation $\displaystyle\frac{\bareps\theta}{n}\Tr(F(y(t))-\theta I)^{-1}=1$ \label{s2-mnmxI}\\
			$X(t) \gets \displaystyle\frac{\bareps\theta(t)}{n}(F(y(t))-\theta(t) I)^{-1}$ 	/* Set the primal solution */ \label{s3.-mnmxI}\\
			$i(t) \gets\argmax_{i} A_i\bullet X(t)$ /* Call the maximization oracle */\label{s3-mnmxI}\\ 
			$\nu(t+1) \gets\displaystyle\frac{X(t)\bullet A_{i(t)} - X(t)\bullet  F(y(t))}{X(t)\bullet A_{i(t)}+X(t)\bullet F(y(t))}$ \label{s4-mnmxI}/*  Compute the error */\\
			$\tau(t+1) \gets \displaystyle\frac{\bareps \theta(t) \nu(t+1)}{4n(X(t)\bullet A_{i(t)}+X(t)\bullet F(y(t)))}$ \label{s5-mnmxI}/*  Compute the step size */\\
			$y(t+1) \gets (1-\tau(t+1))y(t)+\tau(t+1) \b1_{i(t)}$ 	/* Update the dual solution */ \label{s6-mnmxI}\\
			$t \gets t+1$}
		$\eps_{s+1}\gets\frac{\eps_{s}}2$\\
		$s\gets s+1$}
	$\hat X\gets \frac{(1-\eps_{s-1})X(t-1)}{(1+\eps_{s-1})^2\theta(t-1))}$; $\hat y\gets \frac{y(t-1)}{\theta(t-1)}$\label{so-mnmxI}\\
	\Return $(\hat X,\hat y,t)$
	\caption{Logarithmic-potential Algorithm for~\raf{packI}-\raf{coverI}}\label{log-pack-alg}
	%\caption{Pack$(\bareps,\bar y,\bar t)$}
\end{algorithm}
\subsection{Analysis}\label{pack-analysis}
\paragraph{High-level Idea of the Analysis.}~
The proof of $\epsilon$-optimality follows easily from the stopping condition in line~\ref{s1.-mnmxI} of the algorithm, the definition of the ``approximation error" $\nu$ in line~\ref{s4-mnmxI}, and the fact that $X\bullet F(y)$ is a very close approximation of $\lambda_{\min}(F(y(t)))$. The main part of the proof is to bound the number of iterations in the inner while-loop (line~\ref{s1.-mnmxI}). This is done by using a {\it potential function argument}: we define the potential function $\Phi(t):=\Phi(\theta^*(t),F(y(t)))$ and show in Claim~\ref{cl7-mnmxI} that, in each iteration,  the choice of the step size in line~\ref{s5-mnmxI} guarantees that $\Phi(t)$ is increased substantially; on the other hand, by Claim~\ref{cl8-mnmxI}, the potential difference cannot be very large, and the two claims 
together imply that we cannot have many iterations.

\subsubsection{Some Preliminaries}
Up to Claim~\ref{cl10-mnmxI}, we fix a particular iteration $s$ of the outer while-loop in the algorithm.
For simplicity in the following, we will sometimes write $F:=F(y(t))$, $\theta:=\theta(t)$, $\theta^*:=\theta^*(t)$, $X:=X(t)$, $\hat F:=A_{i(t)}$, $\tau:=\tau(t+1)$, $\nu:=\nu(t+1)$, $F':=F(y(t+1))$, and $\theta':=\theta(t+1)$, when the meaning is clear from the context. 
For $H\succ 0$ and $x\in(0,\lambda_{\min}(H))$, define the {\it logarithmic potential function} \cite{GKPV01,NN94}:
\begin{align}\label{log-pot}
\Phi(x,H)=\ln x+\frac{\bareps}{n}\ln\det\big(H-x I\big).
\end{align}
Note that the term $\ln\det\big(H-x I\big)$ forces the value of $x$ to stay away from the ``boundary'' $\lambda_{\min}(H)$, while the term $\ln x$ pushes $x$ towards that boundary; hence, one would expect the maximizer of $\Phi(x,H)$ to be a good approximation of $\lambda_{\min}(H)$ (see Claim~\ref{cl2-mnmxI}).
\begin{claim}\label{cl0-mnmxI}
	If $F(y(t))\succ 0$, then $\theta^*(t)=\argmax_{0<x<\lambda_{\min}(F)}\Phi(x,F(y(t)))$ and $X(t)\succ0$.
\end{claim}	
\begin{proof}
	Note that 
	\begin{align*}
	\frac{d\Phi(x,F)}{dx}&=\frac{1}{x}-\frac{\bareps}{n}\Tr\big((F-x I)^{-1}\big)\quad\text{ and }\quad
	\frac{d^2\Phi(x,F)}{dx^2}=-\frac{1}{x^2}-\frac{\bareps}{n}\Tr\big((F-x I)^{-2}\big).
	\end{align*}
	Thus, if $F\succ 0$, then $\frac{d^2\Phi(x,F)}{dx^2}=-\frac{1}{x^2}-\frac{\bareps}{n}\sum_{j}\frac{1}{(\lambda_j(F)-x)^2}<0$ for all $x\in(0,\lambda_{\min}(F))$. Thus $ \Phi(x,F)$ is {\it strictly concave} in  $x\in(0,\lambda_{\min}(F))$ and hence has a unique maximizer defined by setting $\frac{d\Phi(x,F)}{dx}=0$, giving the definition $\theta^*(t)$ in step~\ref{s2-mnmxI}. Also, by definition of $X$ in
	step~\ref{s3.-mnmxI}, $\lambda_{\min}(X)=\frac{\bareps\theta}{n}(\lambda_{\min}(F)-\theta)^{-1}>0$ (as $\theta<\theta^*<\lambda_{\min}(F)$), implying that $X\succ 0$.
\end{proof}

For $x\in(0,\lambda_{\min}(F))$, let $g(x):=\displaystyle\frac{\bareps x}{n}\Tr(F-x I)^{-1}$.
The following claim shows that our choice of  ${\delta_s}$ guarantees that  $g(\theta)$ is a good approximation of $g(\theta^*)=1$.
\begin{claim}\label{cl0--mnmxI}
	$g(\theta(t))\in(1-\bareps,1)$.
\end{claim}
\begin{proof}
	For $x\in(0,\lambda_{\min}(F))$, we have
	\begin{align}\label{g-}
	\frac{dg(x)}{dx}&=\frac{\bareps}{n}\sum_{j=1}^n\frac{1}{\lambda_j(F)-x}+\frac{\bareps x}{n}\sum_{j=1}^n\frac{1}{(\lambda_j(F)-x)^2}>0,\\
	\frac{d^2g(x)}{dx^2}&=\frac{2\bareps}{n}\sum_{j=1}^n\frac{1}{(\lambda_j(F)-x)^2}+\frac{2\bareps x}{n}\sum_{j=1}^n\frac{1}{(\lambda_j(F)-x)^3}>0.
	\end{align}	 
	Thus, $g(x)$ is monotone increasing and strictly convex in $x$.  As $\theta<\theta^*$, we have $g(\theta)<g(\theta^*)=1$. Moreover, by convexity,
	\begin{align}\label{g---}
	g(\theta)&\ge g(\theta^*)+(\theta-\theta^*)\left.\frac{dg(x)}{dx}\right|_{x=\theta^*}\nonumber\\
	&\ge 1-{\delta_s}\frac{\bareps\theta^*}{n}\sum_{j=1}^n\frac{1}{\lambda_j(F)-\theta^*}-{\delta_s}\frac{\bareps}{n}\sum_{j=1}^n\left(\frac{\theta^*}{\lambda_j(F)-\theta^*}\right)^2\nonumber\tag{$\because (1-\delta_s)\theta^*\le\theta$}\\
	&\ge 1-{\delta_s}-{\delta_s}\frac{\bareps}{n}\left(\sum_{j=1}^n\frac{\theta^*}{\lambda_j(F)-\theta^*}\right)^2\nonumber\tag{by defintition of $\theta^*$ and $\sum_jx_j^2\le\left(\sum_jx_j\right)^2$ for nonnengative $x_j$'s}\\
	&=1-{\delta_s}\left(1+\frac{n}{\bareps}\right)>1-\bareps.\tag{by defintition of $\theta^*$ and $\delta_s$}
	\end{align}	 
\end{proof}
The following two claim show that $\theta(t)\approx X(t) \bullet F(y(t))$ provides a good approximation for $\lambda_{\min}(F(y(t)))$.
\begin{claim}\label{cl2-mnmxI}
	$(1-\bareps)\lambda_{\min}(F(y(t)))<\theta(t)< \frac{\lambda_{\min}(F(y(t)))}{1+\bareps/n}$ and $\frac{\lambda_{\min}(F(y(t)))}{1+\bareps}\le\theta^*(t)\le \frac{\lambda_{\min}(F(y(t)))}{1+\bareps/n}$.
\end{claim}
\begin{proof}
	By Claim~\ref{cl0--mnmxI}, we have
	\begin{align}\label{eq3-mnmxI}
	1-\bareps< \frac{\bareps \theta(t)}{n}\sum_{j=1}^n \frac{1}{\lambda_j(F)-\theta(t)} < 1.
	\end{align}
	The middle term in \raf{eq3-mnmxI} is at least $\frac{\bareps \theta(t)}{n}\frac{1}{\lambda_{\min}(F)-\theta(t)}$ and at most $\frac{\bareps \theta(t)}{n}\frac{n}{\lambda_{\min}(F)-\theta(t)}$, which implies the claim for $\theta(t)$. The claim for $\theta^*(t)$ follows similarly. 
\end{proof}

%\begin{claim}\label{cl2-mnmxI-}
%	$(1-\bareps)\theta^*(t)<\theta(t)<(1+\bareps)\theta^*(t)$.
%\end{claim}	
%
%\begin{proof}
%Immediate from Claim~\ref{cl2-mnmxI}.
%\end{proof}
\begin{claim}\label{cl3-mnmxI}
	$\theta(t)< X(t) \bullet F(y(t)) < (1+\bareps)\theta(t)$.
\end{claim}
\begin{proof}
	\iffalse
	Let $F=U\Lambda U^T$ be the eigen decomposition of $F=F(y(t))$. 	Note that 
	\begin{align*}
	F\bullet (F-\theta I)^{-1}&=U\Lambda U^T \bullet U(\Lambda-\theta I)^{-1}U^T=I\bullet  \Lambda U^TU(\Lambda-\theta I)^{-1}U^TU=I\bullet  \Lambda(\Lambda-\theta I)^{-1}.
	\end{align*}
	It follows from Claim~\ref{cl0--mnmxI} that
	\begin{align*}
	X\bullet F &= \frac{\bareps \theta}{n}\sum_{j=1}^n \frac{\lambda_j(F)}{\lambda_j(F)-\theta}\\
	&= \frac{\bareps \theta}{n}\sum_{j=1}^n \left(1+\frac{\theta}{\lambda_j(F)-\theta}\right)\\
	&= \bareps \theta+\theta\frac{\bareps\theta}{n}\sum_{j=1}^n \frac{1}{\lambda_j(F)-\theta}\\
	&=\bareps \theta+\theta\frac{\bareps\theta}{n}\Tr(F-\theta I)^{-1}\\
	&\in\big(\bareps+(1-\bareps,1)\big)\theta=\big(\theta,(1+\bareps)\theta\big).
	\end{align*}
	\fi
	By the definition of $X$, we have $(F-\theta I)X=\frac{\bareps \theta}{n}I$. It follows from Claim~\ref{cl0--mnmxI} that
	\begin{align*}
	X\bullet F &= \frac{\bareps \theta}{n}\Tr(I)+\theta\Tr(X)\in\big(\bareps+(1-\bareps,1)\big)\theta=\big(\theta,(1+\bareps)\theta\big).
	\end{align*}
\end{proof}

\begin{claim}\label{cl5.-mnmxI}
	$\b1^Ty(t)=1$.
\end{claim}
\begin{proof} 
	This is immediate from the initialization of $y(0)$ in step~\ref{s1-mnmxI} and the update of $y(t+1)$ in step~\ref{s6-mnmxI} of the algorithm.
\end{proof}	

\begin{claim}\label{cl4-mnmxI}
	For all  iterations $t$, except possibly the last, $\nu(t+1),\tau(t+1) \in(0,1)$.
\end{claim}
\begin{proof}
	$\nu(t+1)\ge 0$ as $X\bullet A_{i(t)}\ge X\bullet F$ by Claim~\ref{cl5.-mnmxI}, and  except  possibly for the last iteration, we have $\nu(t+1)>0$.
	Also, $\nu(t+1)\le 1$ by the non-negativity of $X\bullet A_{i(t)}$ and $X\bullet F$, while $\nu(t+1)=1$ implies that $X\bullet F=0$, in contradiction to Claim~\ref{cl3-mnmxI}. 
	
	Note that the definition of $\nu(t+1)$ implies that 
	\begin{align*}
	\tau(t+1)=\frac{\bareps\theta\nu(t+1)(1-\nu(t+1))}{8n X(t)\bullet F(y(t))},
	\end{align*}  and hence,  $\tau(t+1)>0$. 	Moreover, by Claim~\ref{cl3-mnmxI}, $\tau(t+1)<\frac{\bareps}{8n}<1$.   
\end{proof}

\begin{claim}\label{cl5-mnmxI}
	$F(y(t))\succ 0$.
\end{claim}
\begin{proof} 
	This follows by induction on $t'=0,1,\ldots,t$. For $t'=0$, the claim follows from assumption (B-I), which implies that $F(y(0))=\frac1r\bar A\succ0$. Assume now  that $F=F(y(t))\succ 0$. Then for $F'=F(y(t+1))$, we have by step~\ref{s6-mnmxI} of the algorithm that $F'=(1-\tau)F+\tau A_{i(t)}\succ 0$.
\end{proof}	

\begin{claim}\label{cl5-----mnmxI}
	$(F-\theta^*I)^{-1}=\left(\frac{\bareps \theta}{n}I-(\theta^*-\theta)X\right)^{-1}X$.
\end{claim}
\begin{proof}			
	By definition of $X$, we have 	
	\begin{align*}
	(F-\theta^*I)X&=(F-\theta I)X-(\theta^*-\theta)X=\frac{\bareps\theta}{n}I-(\theta^*-\theta)X\\
	\therefore X&=(F-\theta^*I)^{-1}\left(\frac{\bareps\theta}{n}I-(\theta^*-\theta)X\right).
	\end{align*}
\end{proof}

\subsubsection{Number of Iterations}
Define $B=B(t):=\frac{n}{\bareps\theta}\left(\tau X^{1/2}(\hat F-F)X^{1/2}-(\theta^*-\theta)X\right)$.
\begin{claim}\label{cl5---mnmxI}
	$F' - \theta^* I =(F-\theta I)^{1/2}(I+B)(F-\theta I)^{1/2}$.
\end{claim}

\begin{proof}			
	By (the update) step~\ref{s6-mnmxI}, we have
	\begin{align*}
	F' - \theta^* I	&=  (1-\tau)F +\tau \hat F -\theta^* I\nonumber\\
	&= F-\theta I + \tau(\hat F -F)-(\theta^*-\theta)I\nonumber\\
	&= (F-\theta I)^{1/2}\left(I+\frac{n\tau}{\bareps \theta}X^{1/2}(\hat F-F)X^{1/2}-\frac{n}{\bareps \theta}(\theta^*-\theta)X\right)(F-\theta I)^{1/2} \tag{$\because X=\frac{\bareps \theta}{n}(F-\theta I)^{-1}$}.\nonumber
	\label{eq1-mnmxI}
	\end{align*}
\end{proof}

\begin{claim}\label{cl5----mnmxI}
	$\max_j\left|\lambda_j(B)\right|\le\frac12$.
\end{claim}
\begin{proof} 
	By the definition of $B$, we have
	\begin{align*}
	\max_j\left|\lambda_j(B)\right|&=\frac{n}{\bareps\theta}\max_j\left|\lambda_j\left(\tau X^{1/2}(\hat F - F)X^{1/2}-(\theta^*-\theta)X\right)\right| \\
	&= \frac{n}{\bareps\theta}\max_{v:||v||=1} \left|v^T\left(\tau X^{1/2}(\hat F - F)X^{1/2}-(\theta^*-\theta)X\right)v\right|\\
	&\leq  \frac{n}{\bareps\theta}\left(\max_{v:||v||=1}\tau v^TX^{1/2}\hat FX^{1/2}v+\max_{v:||v||=1}\tau v^TX^{1/2}FX^{1/2}v+(\theta^*-\theta)\max_{v:||v||=1}v^TXv\right)\\
	&\le \frac{n\tau}{\bareps\theta}\left(\Tr(X^{1/2}\hat FX^{1/2})+\Tr(X^{1/2}FX^{1/2})\right)+\frac{n{\delta_s}}{(1-{\delta_s})\bareps}\tag{$\because F,\hat F\succeq0,~ |\Tr(X)\le 1$ and $\theta\ge(1-{\delta_s})\theta^*$}\\
	&= \frac{n\tau}{\bareps\theta}\left(X\bullet \hat F + X \bullet F\right)+\frac{n{\delta_s}}{(1-{\delta_s})\bareps}\\
	&=\frac{\nu}{4}+	\frac{\bareps^2 }{32(1- \bareps^3/(32n))} \tag{substituting $\tau$ and ${\delta_s}$}\\
	&< \frac{1}{2}\tag{using $\nu,\bareps\le 1$}.
	\end{align*}
\end{proof}
\begin{claim}\label{cl5--mnmxI}
	$\theta^*(t) < \lambda_{\min}(F(y(t+1)))$.
\end{claim}
\begin{proof} 
	By Claim~\ref{cl5----mnmxI}, $I+B \succeq I	-\frac{1}{2}I=\frac12I, 
	$ and by thus, we get by Claim~\ref{cl5---mnmxI},
	\begin{align*}
	F' - \theta^* I 	&\succeq \frac{1}{2}(F-\theta I)\succ 0. \tag{$\because BZB\succeq 0$ for $B\in\SS^{n}$ and $Z\in\SS^{n}_+$}\nonumber\\
	%&\succ 0. \tag{$\because \theta < \lambda_{\min}(F)$ by Claim~\ref{cl0-mnmxI}}\nonumber
	%\label{eq3--mnmxI}
	\end{align*}
	%	It follows that 
	%	\begin{align*}
	%         F'-\theta^*I&=F'-\theta I+(\theta-\theta^*)I\\
	%         &\succeq\frac{1}{2}(F-\theta I)+(\theta-\theta^*)I\\
	%         &\succeq\frac{1}{2}(F-\theta^* I)-{\delta_s}\theta^*I\tag{$\because (1-{\delta_s})\theta^*\le\theta\le\theta^*$}\\
	%         &=\frac{1}{2}\left(F-(1+2{\delta_s})\theta^*I\right)\\
	%         &\succeq\frac{1}{2}\left(F-\frac{1+2{\delta_s}}{1+\bareps/n}\lambda_{\min}(F)I\right)\tag{by Claim~\ref{cl2-mnmxI}}\\
	%         &\succ\frac{1}{2}\left(F-\lambda_{\min}(F)I\right)\succ 0,\\
	%	\end{align*}
	%	where the last inequality follows by the definition of ${\delta_s}$, as
	%	\begin{align*}
	%	\frac{1+2{\delta_s}}{1+\bareps/n}=\frac{1+n/\bareps+2\bareps}{2+n/\bareps+\bareps/n}<1 \tag{$\because \bareps<\frac{1}{6}$.}
	%	\end{align*}
%	The claim follows.
\end{proof}

\begin{claim}\label{cl6--mnmxI}
	if $\nu>\bareps$, then $\Tr(B)\ge \frac{\nu^2}{8}$.
\end{claim}	
\begin{proof}
	By the definition of $B$,
	\begin{align*}   
	\Tr(B)&=\frac{n}{\bareps\theta}\left(\tau \Tr(X^{1/2}(\hat F-F)X^{1/2})-(\theta^*-\theta)\Tr(X)\right)\\
	&\ge\frac{n}{\bareps\theta}\left(\tau (X\bullet\hat F- X\bullet F)-(\theta^*-\theta)\right)\tag{$\because \Tr(X)\le 1$ by Claim~\ref{cl0--mnmxI}}\\
	&\ge\frac{n}{\bareps\theta}\left(\tau (X\bullet\hat F- X\bullet F)-\frac{{\delta_s}}{1-{\delta_s}}\theta\right)\tag{$\because (1-{\delta_s})\theta^*\le \theta$}\\
	&=\frac{\nu^2}{4}-\frac{\bareps^2}{32(1-\bareps^3/32n))}\tag{by definition of $\tau$ and ${\delta_s}$}\\
	&>\frac{\nu^2}{4}-\frac{\nu^2}{16}>\frac{\nu^2}{8}.\tag{$\because \bareps<\nu\le 1$}
	\end{align*}
\end{proof}

\begin{claim}\label{cl6-mnmxI}
	if $\nu>\bareps$, then $\Tr(B^2)<\frac{\nu^2}{10}$.
\end{claim}	
\begin{proof}
	Write $\hat Y=\tau X^{1/2}\hat FX^{1/2}$ and $Y=X^{1/2}\big(\tau F+(\theta^*-\theta)I\big)X^{1/2}$ and note that both $\hat Y$ and $Y$ are in $\SS_+^n$. It follows by the definition of $B$ that
	\begin{align*}   
	\Tr(B^2)&=\frac{n^2}{\bareps^2\theta^2}\Tr\left((\hat Y-Y)^2\right)\\
	&=\frac{n^2}{\bareps^2\theta^2}\left(\Tr(\hat Y^2)+\Tr(Y^2)-2\Tr(\hat Y Y)\right)\\
	&\le\frac{n^2}{\bareps^2\theta^2}\left(\Tr(\hat Y^2)+\Tr(Y^2)+2\Tr(\hat Y Y)\right)\tag{$\because \hat Y,Y\in\SS^n_+$}\\
	&\le \frac{n^2}{\bareps^2\theta^2}\left(\Tr(\hat Y^2)+\Tr(Y^2)+2\sqrt{\Tr(\hat Y^2)\Tr(Y^2)}\right) \tag{by Cauchy-Schwarz Ineq.}	\\
	&\le\frac{n^2}{\bareps^2\theta^2}\left(\Tr(\hat Y)^2+\Tr(Y)^2+2\Tr(\hat Y )\Tr(Y)\right) \tag{$\because \hat Y,Y\in\SS^n_+$}\\
	&=\frac{n^2}{\bareps^2\theta^2}(\Tr(\hat Y)+\Tr(Y))^2\\
	&=\frac{n^2}{\bareps^2\theta^2}(\tau\Tr(X\hat F)+\tau\Tr(XF)+(\theta^*-\theta)\Tr(X))^2\\
	&\le\frac{n^2}{\bareps^2\theta^2}(\tau X\bullet \hat F+\tau X\bullet F+(\theta^*-\theta))^2\tag{$\because \Tr(X)\le 1$ by Claim~\ref{cl0--mnmxI}}\\
	&\le \frac{n^2}{\bareps^2\theta^2}\left(\tau X\bullet \hat F+\tau X\bullet F+\frac{{\delta_s}}{1-{\delta_s}}\theta\right)^2\tag{$\because (1-{\delta_s})\theta^*\le \theta$}\\
	&=\left(\frac{\nu}{4}+\frac{\bareps^2}{32(1-\bareps^3/(32n))}\right)^2\tag{by definition of $\tau$ and ${\delta_s}$}\\
	&< \left(\frac{\nu}{4}+\frac{\nu^2}{16}\right)^2\tag{$\because \bareps<\nu\le 1$}<\frac{\nu^2}{10}.
	\end{align*}
\end{proof}

Define $\Phi(t):=\Phi(\theta^*(t),F(y(t)))$.

\begin{claim}\label{cl7-mnmxI}
	$\Phi(t+1)-\Phi(t) \geq \frac{\bareps \nu(t+1)^2}{40n}$.
\end{claim}
\begin{proof} 
	Note that Claim~\ref{cl5--mnmxI} implies that $\theta^*$ is feasible to the problem $\max\{\Phi(\xi,F'):0\leq \xi \leq \lambda_{\min}(F')\}$. Thus,
	\begin{align*}
	\Phi(t+1)&=\Phi(\theta^*(t+1),F')
	\geq \ln \theta^* + \frac{\bareps}{n}\ln \det (F'-\theta^* I).\\
	\therefore \Phi(t+1)-\Phi(t) &\geq \frac{\bareps}{n}\left(\ln \det(F'-\theta^* I)-\ln \det(F-\theta^* I)\right)\\
	&\geq \frac{\bareps}{n}\left(\ln \det\big(F'-\theta^*I\big)-\ln \det(F-\theta I)\right)\tag{$\because \theta\le\theta^*$}\\
	&= \frac{\bareps}{n}\ln \det \left(I+B\right)\tag{by Claim~\ref{cl5---mnmxI}}\\
	&=\frac{\bareps}{n}\sum_{j=1}^n\ln  \left(1+\lambda_j(B)\right)\\
	&\geq \frac{\bareps}{n}\sum_{j=1}^n\left( \lambda_j(B)-\lambda_j(B)^2\right)\tag{by Claim~\ref{cl5----mnmxI} and $ \ln(1+z) \geq z-z^2, \forall z\geq -0.5$}\\
	&=\frac{\bareps}{n}\left(\Tr(B)-\Tr(B^2)\right)\\
	&>\frac{\bareps}{8n}\nu^2-\frac{\bareps}{10n}\nu^2\tag{by Claims~\ref{cl6--mnmxI} and \ref{cl6-mnmxI}}\\
	&=\frac{\bareps}{40n}\nu^2.
	\end{align*}
\end{proof}

\begin{claim}\label{cl8-mnmxI}
	For any $t,t'$, $$\Phi(t')-\Phi(t) \leq (1+\bareps)\ln \displaystyle\frac{ X(t) \bullet A_{i(t)}}{(1-\bareps)X(t) \bullet F(y(t))}.$$
\end{claim}
\begin{proof}
	Write $F=F(y(t))$, $\theta^*:=\theta^*(t)$, $\theta:=\theta(t)$, $X:=X(t)$, $F'=F(y(t'))$, $\theta'^{*}:=\theta^*(t')$. Then
	\begin{align*}
	\Phi(t')-\Phi(t) 	&=\ln \frac{\theta'^*}{\theta^*}+\frac{\bareps}{n}\ln \det \left[(F-\theta^*I)^{-1}(F'-\theta'^*I)\right]\\
	&=\ln \frac{\theta'^*}{\theta^*}+\frac{\bareps}{n}\ln \det \left[\left(\frac{\bareps \theta}{n}I-(\theta^*-\theta)X\right)^{-1}X(F'-\theta'^*I)\right]\tag{by Claim~\ref{cl5-----mnmxI}}\\
	&=\ln \frac{\theta'^*}{\theta^*}+\frac{\bareps}{n} \left[\ln \det\left(\frac{\bareps \theta}{n}I-(\theta^*-\theta)X\right)^{-1}+\ln \det \big[X(F'-\theta'^*I)\big]\right]\\
	&\le\ln \frac{\theta'^*}{\theta^*}+\frac{\bareps}{n} \left[\ln \left(\frac{\bareps \theta}{n}-\frac{{\delta_s}\theta}{1-{\delta_s}}\right)^{-n}+\ln \det \big[X(F'-\theta'^*I)\big]\right]\tag{$\because \Tr(X)\le1$  by Claim~\ref{cl0--mnmxI} and $(1-{\delta_s})\theta^*\le\theta$}\\
	&\le\ln \frac{\theta'^*}{\theta^*}+\frac{\bareps}{n} \left[\ln\left(\frac{n}{(1-\bareps)\bareps \theta}\right)^n+\ln\det X(F'-\theta'^*I)\right]\tag{by defintion  of ${\delta_s}$}\\
	&=\ln \frac{\theta'^*}{\theta^*}+\bareps\ln \frac{n}{(1-\bareps)\bareps \theta}+\frac{\bareps}{n}\ln \left[\det X(F'-\theta'^*I)\right]\\
	&=\ln \frac{\theta'^*}{\theta^*}+\bareps\ln \frac{n}{(1-\bareps)\bareps \theta}+\frac{\bareps}{n}\sum_{j=1}^n\ln \lambda_j\big(X(F'-\theta'^*I)\big)\\
	&\leq \ln \frac{\theta'^*}{\theta^*}+\bareps\ln \frac{n}{(1-\bareps)\bareps \theta}+\bareps\ln \left(\frac{1}{n}\sum_{j=1}^n\lambda_j\big(X(F'-\theta'^*I)\big)\right) \tag{by the concavity of $\ln(\cdot)$}\\
	&= \ln \frac{\theta'^*}{\theta^*}+\bareps\ln \frac{n}{(1-\bareps)\bareps \theta}+\bareps\ln \left(\frac{1}{n}\Tr(XF'-\theta'^*X)\right)\\
	&\le \ln \frac{\theta'^*}{\theta^*}+\bareps\ln \frac{n}{(1-\bareps)\bareps \theta}+\bareps\ln \left(\frac{X\bullet F'-\theta'^*(1-\bareps)}{n}\right)\tag{$\because \Tr(X)\ge1-\bareps$ by Claim~\ref{cl0--mnmxI}}\\
	&= \ln \frac{\theta'^*}{\theta^*}+\bareps\ln \frac{1}{(1-\bareps)\bareps \theta}+\bareps\ln \left(X\bullet F'-\theta'^*(1-\bareps)\right)\\
	&\leq \ln \frac{\theta'^*}{\theta^*}+\bareps\ln \frac{1}{(1-\bareps)\bareps \theta}+\bareps\ln \left(\max_{y\in\RR_+^m:~\b1^Ty= 1} X\bullet F(y)-\theta'^*(1-\bareps)\right)\tag{$\because \b1^Ty(t')=1$ by Claim~\ref{cl5.-mnmxI}}\\
	&=\ln \frac{\theta'^*}{\theta^*}+\bareps\ln \frac{1}{(1-\bareps)\bareps \theta}+\bareps\ln \left( X\bullet A_{i(t)}-\theta'^*(1-\bareps)\right)\tag{by defintion of $i(t)$}\\
	&\leq\max_{0\le\xi< X\bullet A_{i(t)}} \left\{\ln \frac{\xi}{(1-\bareps)\theta^*}+\bareps\ln \frac{1}{(1-\bareps)\bareps \theta}+\bareps\ln \left(X\bullet A_{i(t)}-\xi\right)\right\}\\
	&= (1+\bareps)\ln \displaystyle\frac{X\bullet A_{i(t)}}{(1-\bareps^2)\theta}+\ln\frac{\theta}{\theta^*}\tag{$\max(\cdot)$ is achieved at $\xi=\frac{X\bullet A_{i(t)}}{1+\bareps}$}\\
	&\le (1+\bareps)\ln \displaystyle\frac{X\bullet A_{i(t)}}{(1-\bareps^2)\theta}\tag{$\because\theta\le\theta^*$}\\
	&\le (1+\bareps)\ln \displaystyle\frac{X\bullet A_{i(t)}}{(1-\bareps)X \bullet F}.\tag{by Claim~\ref{cl3-mnmxI}}\\
	\end{align*}
\end{proof}
\iffalse
The following claim (similar to the one in \cite{J04}) states that $\Phi(t)$ can be closely approximated by $\ln\theta^*(t)$.
\begin{claim}\label{cl8-mnmxI}
	$(1+\bareps)\ln\theta^*(t)-\bareps\ln\frac{1}{\bareps}\le\Phi(t) \leq (1+\bareps)\ln\theta^*(t)-\frac{\bareps}{n}\ln\frac{1}{\bareps}$.
\end{claim}
\begin{proof}
	By the definition of $\Phi(t)=\Phi(\theta^*,F)$,
	\begin{align}\label{sube1}
	(1+\bareps)\ln\theta^*&=\Phi(t)+\frac{\bareps}{n}\sum_{j=1}^n\ln\left(\frac{\theta^*}{\lambda_j(F)-\theta^*}\right).
	\end{align}
	Using the concavity of the $\ln(\cdot)$ in \raf{sube1},
	\begin{align*}
	(1+\bareps)\ln\theta^*&\le\Phi(t)+\bareps\ln\left(\frac1n\sum_{j=1}^n\frac{\theta^*}{\lambda_j(F)-\theta^*}\right)\\
	&\le\Phi(t)+\bareps\ln\left(\frac{\bareps}n\sum_{j=1}^n\frac{\theta^*}{\lambda_j(F)-\theta^*}\right)-\bareps\ln\bareps\\
	&=\Phi(t)-\bareps\ln\bareps\tag{by defintion of $\theta^*$},
	\end{align*}
	from which the lower bound in the claim follows. For the upper bound, we have by \raf{sube1}, 
	\begin{align*}
	(1+\bareps)\ln\theta^*&\ge\Phi(t)+\frac{\bareps}n\ln\left(\frac{\theta^*}{\lambda_{\min}(F)-\theta^*}\right)\\
	&\ge\Phi(t)+\frac{\bareps}n\ln\frac1{\bareps}\tag{by Claim~\ref{cl2-mnmxI}}.
	\end{align*}
\end{proof}
\fi
%\begin{claim}\label{cl9-mnmxI} $X(t) \bullet A_{i(t)}\leq m X(t) \bullet F(y(0))$.\end{claim}
%\begin{proof}
%	\begin{align*}
%	X(t)\bullet A_{i(t)}&=\max_{i}X(t) \bullet A_i 	\leq \sum_iX (t)\bullet A_i=m X(t)\bullet F(y(0)). 
%	\end{align*}
%\end{proof}

Recall by assumption~(B-I) that $\bar A:=\sum_{i=1}^rA_i\succ0$. 
\begin{claim}\label{cl9-mnmxI} $\frac{X(0) \bullet A_{i(0)}}{X(0) \bullet F(y(0))}\leq \psi:=\frac{r\cdot\lambda_{\max}(A_{i(0)})}{\lambda_{\min}(\bar A)}\le \frac{r\cdot\max_i\lambda_{\max}(A_i)}{\lambda_{\min}(\bar A)}\le n\tau 2^\cL$.\end{claim}
\begin{proof}
	Let $X(0)=\sum_{j=1}^n\lambda_ju_ju_j^T$ be the spectral decomposition of $X(0)$. Then, 
	\begin{align*}
	X(0)\bullet A_{i(0)}&=\sum_{j=1}^n\lambda_jA_{i(0)}\bullet u_ju_j^T	\leq\sum_{j=1}^n\lambda_j\lambda_{\max}(A_{i(0)})= \lambda_{\max}(A_{i(0)})\cdot\Tr(X(0))\\
	X(0)\bullet  F(y(0))&=\sum_{j=1}^n\lambda_j F(y(0))\bullet u_ju_j^T	\geq\frac1r\sum_{j=1}^n\lambda_j\lambda_{\min}(\bar A)=\frac1r\lambda_{\min}(\bar A)\cdot\Tr(X(0)). 
	\end{align*}
	The claim follows.
\end{proof}

\begin{claim}\label{cl10-mnmxI}
	The algorithm terminates in at most $O\big(n\log \psi+\frac{n}{\epsilon^2}\big)$ iterations.
\end{claim}

\begin{proof}
	Let $t_{-1}=-1$ and, for $s=0,1,2,\ldots$, let $t_s$ be the smallest $t$ such that $\nu(t+1)\le2^{-(s+1)}$ (so $t_s+1$ is the value of $t$ at which the iteration $s+1$ of the outer while-loop starts). Then for $t=t_{s-1}+1,\ldots,t_{s}-1$, we have $\nu(t+1)>2^{-(s+1)}=\eps_s$. Hence, for $s=0$,
	\begin{align*}
	\frac{\eps_0^3 t_0}{40n}&<\Phi(t_0)-\Phi(0)	\tag{by Claim~\ref{cl7-mnmxI}} \\ &\le(1+\eps_0)\ln\frac{X(0)\bullet A_{i(0)}}{(1-\eps_0)X(0)\bullet F(y(0))}	\tag{by Claim~\ref{cl8-mnmxI}} \\ 
	&\le(1+\eps_0)\ln \frac{\psi}{(1-\eps_0)}	\tag{by Claim~\ref{cl9-mnmxI}}.
	\end{align*}
	Setting $\eps_0=\frac{1}{2}$ in the last series of inequalities we get 
	\begin{align}\label{est1}
	t_0&<480n\ln(2\psi)=O(n\log \psi).
	\end{align}
	Now consider $s\ge 1$: 
	\begin{align*}
	\frac{\eps_s^3 (t_s-t_{s-1})}{40n}&<\Phi(t_s)-\Phi(t_{s-1})	\tag{by Claim~\ref{cl7-mnmxI}} \\ &\le(1+\eps_s)\ln\frac{X(t_{s-1})\bullet A_{i(t_{s-1})}}{(1-\eps_s)X(t_{s-1})\bullet F(y(t_{s-1}))}	\tag{by Claim~\ref{cl8-mnmxI}} \\ 
	&=
	(1+\eps_s)\ln \frac{1+\nu(t_{s-1}+1)}{(1-\eps_s)\big(1-\nu(t_{s-1}+1)\big)}	\tag{by definition of ~$\nu(t_{s-1}+1)$}\\
	&\le(1+\eps_s)\ln \frac{1+2\eps_s}{(1-\eps_s)(1-2\eps_s)}	\tag{$\because \nu(t_{s-1}+1)\le2^{-s}=2\eps_s$}\\
	&\le(1+\eps_s)\ln(1+12\eps_s)\tag{$\because\eps_s\le\frac14$}
	\le15\eps_s.
	\end{align*}
	Setting $\eps_s=\frac{1}{2^{s+1}}$ in the last series of inequalities we get 
	\begin{align}\label{est2}
	t_s-t_{s-1}&<\frac{600n}{\eps_s^2}=O(n/\eps_s^2).
	\end{align}
	Summing \raf{est1}, and \raf{est2} over $s=1,2,\ldots,\lceil\log\frac{1}{\epsilon}\rceil$, we get the claim.
\end{proof}
\begin{remark}\label{r1}
	If we do not insist on a sparse dual solution, then we can use the initialization $y(0)\gets\frac1m\b1$ in step~\ref{s1-mnmxI} in Algorithm~\ref{log-pack-alg}, where  $\b1$ is the $m$-dimensional vector of all ones, and replace $\psi$ in Claim~\ref{cl9-mnmxI}, and hence in the running time in Claim~\ref{cl10-mnmxI}, by $m$. 
\end{remark}
\subsubsection{Primal Dual Feasibility and Approximate Optimality  }
Let $t_f+1$ be the value of $t$ when the algorithm terminates and $s_f+1$ be the value of $s$ at termination. For simplicity, we write  $s=s_f$.

\begin{claim}\label{cl11-mnmxI}(Primal feasibility).
	$\hat X\succ 0$ and $\max_i  A_i\bullet \hat X\leq 1$.
\end{claim}
\begin{proof}
	The first claim is immediate from Claim~\ref{cl0-mnmxI}. To see the second claim, we use the definition of $\nu(t_f)$ and the termination condition in line~\ref{s1.-mnmxI} (which is also satisfied even if $X(t_f)\bullet A_{i(t_f)} - X(t_f)\bullet  F(y(t_f))=0$):
	\begin{align*}
	\frac{X(t_f)\bullet A_{i(t_f)} - X(t_f)\bullet  F(y(t_f))}{X(t_1)\bullet A_{i(t_f)} + X(t_f)\bullet  F(y(t_f))} &\leq \bareps.\\
	\therefore (1+\bareps)X(t_f)\bullet F(y(t_f)) &\geq (1-\bareps)X(t_f)\bullet  A_{i(t_f)}\\
	&= (1-\bareps)\max_i X(t_f)\bullet A_i \tag{by the defintition of $i(t_f)$}\\
	\therefore (1+\bareps)^2 \theta(t_f) &\geq (1-\bareps)\max_i X(t_f)\bullet A_i \tag{$\because X(t_f) \bullet F(y(t_f)) \le (1+\bareps)\theta(t_f)$ by Claim~\ref{cl3-mnmxI}}.
	\end{align*}
	The claim follows by the definition of $\hat X$ in step~\ref{so-mnmxI} of the algorithm.
\end{proof}
\begin{claim}(Dual feasibility).\label{cl12-mnmxI}
	$\hat y\ge0$ and $F(\hat y)\succ I$.
\end{claim}
\begin{proof}
	The fact that $\hat y\ge 0$ follows from the initialization of $y(0)$ in step~\ref{s1-mnmxI}, Claim~\ref{cl4-mnmxI},  and the update of $y(t+1)$ in step~\ref{s6-mnmxI}. For the other claim, we have \begin{align*}
	\lambda_{\min}\big(F(\hat y)\big)&=\frac{1}{\theta(t_f)}\lambda_{\min}\big(F(y(t_f))\big)
	\ge1+\frac{\bareps}n.\tag{by Claim~\ref{cl2-mnmxI}}
	\end{align*} 
\end{proof}

\begin{claim}(Approximate optimality).\label{cl13-mnmxI}
	$I\bullet \hat X\ge\left(\frac{1-\bareps}{1+\bareps}\right)^2\b1^T\hat y$.
\end{claim}
\begin{proof}
	By Claim~\ref{cl0--mnmxI}, we have $\Tr(X(t_f))\ge1-\bareps$, and by Claim~\ref{cl5.-mnmxI}, we have $\b1^Ty(t_f)=1$. The claim follows by the definition of $\hat X$ and $\hat y$ in step~\ref{so-mnmxI}.   
\end{proof}	

\begin{remark}\label{r2}
	Suppose that in step~\ref{s3-mnmxI} of Algorithm~\ref{log-pack-alg}, we instead define $i(t)$ to be an index $i\in[m]$ such that $A_i\bullet X(t)\ge 1-\eps_s$, and we are guaranteed that such index exists in each iteration of the algorithm. Then the dual solution $\hat y$ satisfies: $\b1^T\hat y\le 1+O(\epsilon)$. Indeed, the proof of Claim~\ref{cl11-mnmxI} can be easily modified to show that $\theta(t_f)\ge\frac{(1-\eps_{s_f})^2}{(1+\eps_{s_f})^2}$, which combined with the definition of $\hat y$ in step~\ref{so-mnmxI} of the algorithm implies the claim.
\end{remark}	
\subsubsection{Running Time per Iteration}

\paragraph{Computing $\theta(t)$.} Given $F:=F(y(t))\succ 0$, we first compute an approximation $\widetilde\lambda$ of $\lambda_{\min}(F)$ using Lanczos' algorithm with a random start~\cite{KW99}.
\begin{lemma}[\cite{KW99}]\label{KW}
	Let $M\in\SS_+^n$ be a positive semidefinite matrix with $N$ non-zeros and $\gamma\in(0,1)$ be a given constant. Then there is a randomized algorithm that computes, with high (i.e., $1-o(1)$)  probability a unit vector $v\in\RR^n$ such that $v^TMv\ge(1-\gamma)\lambda_{\max}(M)$. The algorithm takes $O\big(\frac{\log n}{\sqrt{\gamma}}\big)$ iterations, each requiring $O(N)$ arithmetic operations.   
\end{lemma} 
By Claim~\ref{cl2-mnmxI}, we need $\widetilde\lambda$ to lie in the range $[\frac{\lambda_{\min}(F)}{1+\eps_s/n},\lambda_{\min}(F)]$.
To obtain $\widetilde\lambda$, we  may apply the above lemma with $M:=F^{-1}$ and $\gamma:=\frac{\eps_s}{2n}$. Then in $O\big(\sqrt{\frac{n}{\eps_s}}\log n\big)$ iterations we get $\widetilde\lambda:=\frac{1-\gamma}{v^TF^{-1}v}$ satisfying our requirement. However, we can save (roughly) a factor of $\sqrt{n}$ in the running time by using, instead, $M:=F^{-n}$ and $\gamma:=\frac{\eps_s}{2}$. Let $v$ be the vector obtained from Lemma~\ref{KW}, and set $\widetilde \lambda:=\left(\frac{1-\gamma}{v^TF^{-n}v}\right)^{1/n}$. Then, as $\lambda_{\max}(M)\ge v^TMv\ge(1-\gamma)\lambda_{\max}(M)$, and $\lambda_{\min}(F)=\lambda_{\max}(F^{-n})^{-1/n}$, we get \begin{align}\label{ee--I}
\frac{\lambda_{\min}(F)}{1+\eps_s/n}\le(1-\gamma)^{1/n}\lambda_{\min}(F)\le\widetilde\lambda \le\lambda_{\min}(F).
\end{align}
Note that we can compute $F^{-n}$ in $O(n^{\omega}\log n)$, where $w$ is the {\it exponent of matrix multiplication}. Thus, the overall running time for computing $\widetilde\lambda$ is $O(n^{\omega}\log n+\frac{n^2\log n}{\sqrt{\eps_s}})$.

Given $\widetilde\lambda$, we know by Claim~\ref{cl2-mnmxI} and \raf{ee--I} that $\theta^*(t)\in[\frac{\widetilde\lambda}{1+\eps_s},\widetilde\lambda]$. Then we can apply binary search to find $\theta(t):=\theta^*(t)_{\delta_s}$ as follows. Let $\theta_k=\frac{\widetilde\lambda}{1+\eps_s}(1+\delta_s)^k$, for $k=0,1,\ldots,K:=\lceil\frac{2\ln(1+\eps_s)}{\delta_s}\rceil$, and note that $\theta_L\ge\widetilde\lambda$. Then we do binary search on the exponent $k\in\{0,1\ldots,K\}$; each step of the search evaluates $g(\theta_k):=\displaystyle\frac{\bareps \theta_\ell}{n}\Tr(F-\theta_k I)^{-1}$, and depending on whether this value is less than or at least $1$, the value of $k$ is increased or decreased, respectively. The search stops when the search interval $[\ell,u]$ has $u\le \ell+1$, in which case we set $\theta(t)=\theta_{\ell}$; the number of steps until this happens is $O(\log K)=O(\log\frac{1}{\delta_s})=O(\log\frac n{\eps_s})$. By the monotonicity of $g(x)$ (in the interval $[0,\lambda_{\min}(F)]$), and the property of binary search, we know that $\theta^*\in[\theta_\ell,\theta_u]$. Thus, by the stopping criterion, $$\theta(t)=\theta_\ell\le\theta^*(t)\le\theta_u\le\theta_{\ell+1}=(1+\delta_s)\theta_\ell,$$
implying that $(1-\delta_s)\theta^*(t)\le\theta(t)\le\theta^*(t)$.
Since evaluating $g(\theta_\ell)$ takes $O(n^\omega)$, the overall running time for the binary search procedure is $O(n^{\omega}\log\frac n{\eps_s})$, and hence the total time needed for for computing $\theta(t)$ is $O(n^{\omega}\log \frac n{\epsilon}+\frac{n^2\log n}{\sqrt{\epsilon}})$.

All other steps of the algorithm inside the inner while-loop can be done in $O(\cT+n^2)$ time, where $\cT$ is the time taken by a single call to the oracle Max$(X(t))$ in step~\ref{s3-mnmxI} of the algorithm. Thus, in view of Claim~\ref{cl10-mnmxI}, we obtain the following  result.
\begin{theorem}\label{thm:pack-main}
	For any $\epsilon>0$, Algorithm~\ref{log-pack-alg} outputs an $O(n\log \psi+\frac{n}{\epsilon^2})$-sparse $O(\epsilon)$-optimal primal-dual pair in time\footnote{$\widetilde O(\cdot)$ hides polylogarithmic factors in $n$ and $\frac1\epsilon$.} $O\big((n\log \psi+\frac{n}{\epsilon^2})(n^{\omega}\log \frac n{\epsilon}+\frac{n^2\log n}{\sqrt{\epsilon}}+\cT)\big)=\widetilde O\big(\frac{n^{\omega+1}\log \psi}{\epsilon^{2.5}}+\frac{n\cT\log\psi}{\epsilon^2}\big)$. 
\end{theorem}	

%\subsubsection{Improvement for Low-rank Sparse Matrices}
%We assume in this section that $A_i=UU^T$, for some matrix  $U\in\RR^{n\times d}$ having at most $\sqrt{r}$ non-zeros. The idea is to maintain the matrices $\big(F(y(t))-\theta(t) I\big)^{-1}$ for all possible values of $\theta(t)$, and then use low-rank updates to obtain $(F(y(t+1))-\theta(t+1) I)^{-1}$ from them.
%
%Using the Woodbury matrix identity \cite{}:
%
%\begin{align}\label{WMI}
%\big(F(y(t+1)-\theta_\ell\big)^{-1}=\frac{1}{1-\tau(t+1)}\left(\big(F(y(t)-\theta_\ell\big)^{-1}\right)
%\end{align}
\subsection{Algorithm for~\raf{packII}-\raf{coverII}}\label{sec:log-cover}
In this section we give an algorithm for finding a sparse $O(\epsilon)$-optimal primal-dual solution for~\raf{packII}-\raf{coverII}.

For numbers $x\in\RR_+$ and $\delta\in(0,1)$,  a $\delta$-(upper) approximation $x^\delta$ of $x$ is a number such that  $x\le x^\delta<(1+\delta)x$.

The algorithm is shown as Algorithm~\ref{log-cover-alg}. 
%The main while-loop (step~\ref{s1.-mnmxI}) is embedded within a sequence of scaling phases, in which each phase starts from the vector $y(t)$ computed in the previous phase and uses double the accuracy. The algorithm stops when the scaled accuracy $\eps_s$ drops below the desired accuracy $\epsilon\in(0,1/2)$.

%\begin{algorithm}[H]
%	\SetAlgoLined
%	$s \gets 0$; $\epsilon^0\gets 1$; $t^0\gets 0$; $y^0 \gets\frac1m\b1$\label{}\\
%\While{$\epsilon^s>\epsilon$}{
%   $(X^{s+1},y^{s+1},t^{s+1})\gets$Pack$(\epsilon^s,y^s,t^s)$\\
%   $\epsilon^{s+1}\gets\epsilon^s/2$\\
%   $s\gets s+1$
%}
%\caption{Logarithmic-potential Packing Algorithm}\label{log-pack-alg-main}
%\end{algorithm}

\begin{algorithm}[H]
	\SetAlgoLined
	$s \gets 0$; $\eps_0\gets \frac14$; $t\gets 0$; 	 $\nu(0) \gets 1$; $y(0) \gets\b1_i$ (for an arbitrary $i\in[m]$)\label{s1-mnmxII}\\
	\While{$\bareps>\epsilon$}{\label{s0.-mnmxII}
		$\delta_s\gets\frac{\eps_s^3}{32n}$\\
		\While{$\nu(t) >\bareps$}{ \label{s1.-mnmxII}
			$\theta(t)\gets\theta^*(t)^{\delta_s}$, where $\theta^*(t)$ is the smallest positive number root of the equation $\displaystyle\frac{\bareps\theta}{n}\Tr(\theta I-F(y(t)))^{-1}=1$ \label{s2-mnmxII}\\
			$X(t) \gets \displaystyle\frac{\bareps\theta(t)}{n}(\theta(t) I-F(y(t)))^{-1}$ 	/* Set the primal solution */ \label{s3.-mnmxII}\\
			$i(t) \gets\argmin_{i} A_i\bullet X(t)$ /* Call the minimization oracle */ \label{s3-mnmxII}\\ 
			$\nu(t+1) \gets\displaystyle\frac{ X(t)\bullet  F(y(t)) -X(t)\bullet A_{i(t)}}{X(t)\bullet A_{i(t)}+X(t)\bullet F(y(t))}$ \label{s4-mnmxII}/*  Compute the error */\\
			$\tau(t+1) \gets \displaystyle\frac{\bareps \theta(t) \nu(t+1)}{4n(X(t)\bullet A_{i(t)}+X(t)\bullet F(y(t)))}$ \label{s5-mnmxII}/*  Compute the step size */\\
			$y(t+1) \gets (1-\tau(t+1))y(t)+\tau(t+1) \b1_{i(t)}$ 	/* Update the dual solution */ \label{s6-mnmxII}\\
			$t \gets t+1$}
		$\eps_{s+1}\gets\frac{\eps_{s}}2$\\
		$s\gets s+1$}
	$\hat X\gets \frac{(1+\eps_{s-1})X(t-1)}{(1-2\eps_{s-1})^2\theta(t-1))}$; $\hat y\gets \frac{y(t-1)}{\theta(t-1)}$\label{so-mnmxII}\\
	\Return $(\hat X,\hat y,t)$
	\caption{Logarithmic-potential Algorithm for~\raf{packII}-\raf{coverII}}\label{log-cover-alg}
	%\caption{Pack$(\bareps,\bar y,\bar t)$}
\end{algorithm}
\subsection{Analysis}
\subsubsection{Some Preliminaries}
Up to Claim~\ref{cl10-mnmxII}, we fix a particular iteration $s$ of the outer while-loop in the algorithm.
For simplicity in the following, we will sometimes write $F:=F(y(t))$, $\theta:=\theta(t)$, $\theta^*:=\theta^*(t)$, $X:=X(t)$, $\hat F:=A_{i(t)}$, $\tau:=\tau(t+1)$, $\nu:=\nu(t+1)$, $F':=F(y(t+1))$, and $\theta':=\theta(t+1)$, when the meaning is clear from the context. 
For $H\succ 0$ and $x\in(0,\lambda_{\min}(H))$, define the following logarithmic potential function:
\begin{align}\label{log-pot-}
\Phi(x,H)=\ln x-\frac{\bareps}{n}\ln\det\big(x I-H\big).
\end{align}
\begin{claim}\label{cl0-mnmxII}
	If $\lambda_{\max}(F)> 0$, then $\theta^*(t)=\argmin_{x>\lambda_{\max}(F)}\Phi(x,F(y(t)))$ and $X(t)\succ0$.
\end{claim}	
\begin{proof}
	Note that 
	\begin{align*}
	\frac{d\Phi(x,F)}{dx}&=\frac{1}{x}-\frac{\bareps}{n}\Tr\big((x I-F)^{-1}\big)\quad\text{ and }\quad
	\frac{d^2\Phi(x,F)}{dx^2}=-\frac{1}{x^2}+\frac{\bareps}{n}\Tr\big((x I-F)^{-2}\big).
	\end{align*}
	Note that $ \Phi(x,F)$ is not convex in  $x\in(\lambda_{\max}(F),+\infty)$, but has a unique minimizer in this interval, defined by setting $\frac{d\Phi(x,F)}{dx}=0$, giving the definition $\theta^*(t)$ in step~\ref{s2-mnmxII} of Algorithm~\ref{log-cover-alg}. (Indeed, $\frac{d\Phi(x,F)}{dx}<0$ for $\lambda_{\max}(F)<x<\theta^*(t)$, while $\frac{d\Phi(x,F)}{dx}>0$ for $x>\theta^*(t)$.) Also, by definition of $X$ in
	step~\ref{s3.-mnmxI}, $\lambda_{\min}(X)=\frac{\bareps\theta}{n}(\theta-\lambda_{\min}(F))^{-1}>0$ (as $\theta\ge\theta^*>\lambda_{\max}(F)\ge\lambda_{\min}(F)$), implying that $X\succ 0$.
\end{proof}

For $x\in(\lambda_{\max}(H),+\infty)$, let $g(x):=\displaystyle\frac{\bareps x}{n}\Tr(xI-H)^{-1}$.
The following claim shows that our choice of  ${\delta_s}$ guarantees that  $g(\theta)$ is a good approximation of $g(\theta^*)=1$.
\begin{claim}\label{cl0--mnmxII}
	$g(\theta(t))\in(1-\bareps,1]$.
\end{claim}
\begin{proof}
	For $x\in(\lambda_{\max}(H),+\infty)$, we have
	\begin{align}\label{g-II}
	\frac{dg(x)}{dx}&=\frac{\bareps}{n}\sum_{j=1}^n\frac{1}{x-\lambda_j(F)}-\frac{\bareps x}{n}\sum_{j=1}^n\frac{1}{(x-\lambda_j(F))^2}=-\frac{\bareps}{n}\sum_{j=1}^n\frac{\lambda_j(F)}{(x-\lambda_j(F))^2}<0,\\
	\frac{d^2g(x)}{dx^2}&=-\frac{2\bareps}{n}\sum_{j=1}^n\frac{1}{(\lambda_j(F)-x)^2}+\frac{2\bareps x}{n}\sum_{j=1}^n\frac{1}{(\lambda_j(F)-x)^3}=\frac{2\bareps}{n}\sum_{j=1}^n\frac{\lambda_j(F)}{(x-\lambda_j(F))^3}>0.
	\end{align}	 
	Thus, $g(x)$ is monotone decreasing and strictly convex in $x$.  As $\theta\ge\theta^*$, we have $g(\theta)\le g(\theta^*)=1$. Moreover, by convexity,
	\begin{align}\label{g--}
	g(\theta)&\ge g(\theta^*)+(\theta-\theta^*)\left.\frac{dg(x)}{dx}\right|_{x=\theta^*}\nonumber\\
	&\ge 1+{\delta_s}\frac{\bareps\theta^*}{n}\sum_{j=1}^n\frac{1}{\theta^*-\lambda_j(F)}-{\delta_s}\frac{\bareps}{n}\sum_{j=1}^n\left(\frac{\theta^*}{\theta^*-\lambda_j(F)}\right)^2\nonumber\tag{$\because \theta<(1+\delta_s)\theta^*$ and $\left.\frac{dg(x)}{dx}\right|_{x=\theta^*}<0$}\\
	&\ge 1+{\delta_s}-{\delta_s}\frac{\bareps}{n}\left(\sum_{j=1}^n\frac{\theta^*}{\theta^*-\lambda_j(F)}\right)^2\nonumber\tag{by defintition of $\theta^*$ and $\sum_jx_j^2\le\left(\sum_jx_j\right)^2$ for nonnengative $x_j$'s}\\
	&=1+{\delta_s}-\delta_s\frac{n}{\bareps}>1-\bareps.\tag{by defintition of $\theta^*$ and $\delta_s$}
	\end{align}	 
\end{proof}
The following claim shows that $\theta(t)$ provides a good approximation for $\lambda_{\max}(F(y(t)))$.
\begin{claim}\label{cl2-mnmxII}
	$\frac{\lambda_{\max}(F(y(t)))}{1-\bareps/n}<\theta(t)\le \frac{(1-\bareps)\lambda_{\max}(F(y(t)))}{1-2\bareps} $ and $\frac{\lambda_{\max}(F(y(t)))}{1-\bareps/n}\le\theta^*(t)\le \frac{\lambda_{\max}(F(y(t)))}{1-\bareps}$.
\end{claim}
\begin{proof}
	By Claim~\ref{cl0--mnmxII}, we have
	\begin{align}\label{eq3-mnmxII}
	1-\bareps< \frac{\bareps \theta(t)}{n}\sum_{j=1}^n \frac{1}{\theta(t)-\lambda_j(F)} \le 1.
	\end{align}
	The middle term in \raf{eq3-mnmxII} is at least $\frac{\bareps \theta(t)}{n}\frac{1}{\theta(t)-\lambda_{\max}(F)}$ and at most $\frac{\bareps \theta(t)}{n}\frac{n}{\theta(t)-\lambda_{\max}(F)}$, which implies the claim for $\theta(t)$. The claim for $\theta^*(t)$ follows similarly. 
\end{proof}

%\begin{claim}\label{cl2-mnmxI-}
%	$(1-\bareps)\theta^*(t)<\theta(t)<(1+\bareps)\theta^*(t)$.
%\end{claim}	
%
%\begin{proof}
%Immediate from Claim~\ref{cl2-mnmxI}.
%\end{proof}
\begin{claim}\label{cl3-mnmxII}
	$(1-2\bareps)\theta(t)< X(t) \bullet F(y(t)) \le (1-\bareps)\theta(t)$.
\end{claim}
\begin{proof}
	By the definition of $X$, we have $(\theta I-F)X=\frac{\bareps \theta}{n}I$. It follows from Claim~\ref{cl0--mnmxII} that
	\begin{align*}
	X\bullet F &= \theta\Tr(X)-\frac{\bareps \theta}{n}\Tr(I)\in\big((1-\bareps,1]-\bareps\big)\theta=\big((1-2\bareps)\theta,(1-\bareps)\theta\big].
	\end{align*}
\end{proof}

\begin{claim}\label{cl5.-mnmxII}
	$\b1^Ty(t)=1$.
\end{claim}
\begin{proof} 
	This is immediate from the initialization of $y(0)$ in step~\ref{s1-mnmxII} and the update of $y(t+1)$ in step~\ref{s6-mnmxII} of the algorithm.
\end{proof}	

\begin{claim}\label{cl4-mnmxII}
	For all  iterations $t$ in the while-loop, except possibly the last, $\nu(t+1),\tau(t+1) \in(0,1)$.
\end{claim}
\begin{proof}
	$\nu(t+1)\ge 0$ as $X\bullet A_{i(t)}\le X\bullet F$ by Claim~\ref{cl5.-mnmxII}, and  except  possibly for the last iteration, we have $\nu(t+1)>0$.
	Also, $\nu(t+1)\le 1$ by the non-negativity of $X\bullet A_{i(t)}$ and $X\bullet F$, while $\nu(t+1)=1$ implies that $X\bullet A_{i(t)}=0$, in contradiction to the assumption that $A_{i(t)}\ne0$ (as $X\succ 0$ by Claim~\ref{cl0-mnmxII}). 
	
	Note that the definition of $\nu(t+1)$ implies that 
	\begin{align*}
	\tau(t+1)=\frac{\bareps\theta\nu(t+1)(1+\nu(t+1))}{8n X(t)\bullet F(y(t))},
	\end{align*}  and hence,  $\tau(t+1)>0$. 	Moreover, by Claim~\ref{cl3-mnmxII}, $\tau(t+1)<\frac{\bareps}{2n}<1$.   
\end{proof}

\begin{claim}\label{cl5-mnmxII}
	$\lambda_{\max}(F(y(t)))>0$.
\end{claim}
\begin{proof} 
	This follows by induction on $t'=0,1,\ldots,t$. For $t'=0$, the claim follows from the assumption that $A_i\ne 0$ for all $i$. Assume now  that $F=F(y(t))\ne 0$. Then for $F'=F(y(t+1))$, we have by step~\ref{s6-mnmxII} of the algorithm that $F'=(1-\tau)F+\tau A_{i(t)}\ne 0$. As $F'\succeq0$, we get $\lambda_{\max}(F')>0$. 
\end{proof}	

\begin{claim}\label{cl5-----mnmxII}
	$(\theta^*I-F)^{-1}=\left(\frac{\bareps \theta}{n}I-(\theta-\theta^*)X\right)^{-1}X$.
\end{claim}
\begin{proof}			
	By definition of $X$, we have 	
	\begin{align*}
	(\theta^*I-F)X&=(\theta I-F)X-(\theta-\theta^*)X=\frac{\bareps\theta}{n}I-(\theta-\theta^*)X\\
	\therefore X&=(\theta^*I-F)^{-1}\left(\frac{\bareps\theta}{n}I-(\theta-\theta^*)X\right).
	\end{align*}
\end{proof}

\subsubsection{Number of Iterations}
Define $B=B(t):=\frac{n}{\bareps\theta}\left(\tau X^{1/2}( F-\hat F)X^{1/2}-(\theta-\theta^*)X\right)$.
\begin{claim}\label{cl5---mnmxII}
	$\theta^* I-F'  =(\theta I-F)^{1/2}(I+B)(\theta I-F)^{1/2}$.
\end{claim}

\begin{proof}			
	By (the update) step~\ref{s6-mnmxII}, we have
	\begin{align*}
	\theta^* I-F'&=  \theta^* I-(1-\tau)F -\tau \hat F\nonumber\\
	&= \theta I-F + \tau( F -\hat F)-(\theta-\theta^*)I\nonumber\\
	&= (\theta I-F)^{1/2}\left(I+\frac{n\tau}{\bareps \theta}X^{1/2}(F-\hat F)X^{1/2}-\frac{n}{\bareps \theta}(\theta-\theta^*)X\right)(\theta I-F)^{1/2} \tag{$\because X=\frac{\bareps \theta}{n}(\theta I-F)^{-1}$}.\nonumber
	\label{eq1-mnmxII}
	\end{align*}
\end{proof}

\begin{claim}\label{cl5----mnmxII}
	$\max_j\left|\lambda_j(B)\right|\le\frac12$.
\end{claim}
\begin{proof} 
	By the definition of $B$, we have
	\begin{align*}
	\max_j\left|\lambda_j(B)\right|&=\frac{n}{\bareps\theta}\max_j\left|\lambda_j\left(\tau X^{1/2}(F - \hat F)X^{1/2}-(\theta-\theta^*)X\right)\right| \\
	&= \frac{n}{\bareps\theta}\max_{v:||v||=1} \left|v^T\left(\tau X^{1/2}(F -\hat  F)X^{1/2}-(\theta-\theta^*)X\right)v\right|\\
	&\leq  \frac{n}{\bareps\theta}\left(\max_{v:||v||=1}\tau v^TX^{1/2}FX^{1/2}v+\max_{v:||v||=1}\tau v^TX^{1/2}\hat FX^{1/2}v+(\theta-\theta^*)\max_{v:||v||=1}v^TXv\right)\\
	&< \frac{n\tau}{\bareps\theta}\left(\Tr(X^{1/2} FX^{1/2})+\Tr(X^{1/2}\hat FX^{1/2})\right)+\frac{n{\delta_s}}{\bareps}\tag{$\because F,\hat F\succeq0,~|\Tr(X)\le 1$ and $\theta^*\le\theta<(1+\delta_s)\theta^*$}\\
	&= \frac{n\tau}{\bareps\theta}\left(X\bullet \hat F + X \bullet F\right)+\frac{n{\delta_s}}{\bareps}\\
	&=\frac{\nu}{4}+	\frac{\bareps^2 }{32} \tag{substituting $\tau$ and ${\delta_s}$}\\
	&< \frac{1}{2}\tag{using $\nu,\bareps\le 1$}.
	\end{align*}
\end{proof}
\begin{claim}\label{cl5--mnmxII}
	$\theta^*(t) > \lambda_{\max}(F(y(t+1)))$.
\end{claim}
\begin{proof} 
	By Claim~\ref{cl5----mnmxII}, $I+B \succeq I	-\frac{1}{2}I=\frac12I, 
	$ and by thus, we get by Claim~\ref{cl5---mnmxII},
	\begin{align}
	\theta^* I -F'	&\succeq \frac{1}{2}(\theta I-F)\succ 0. \tag{$\because BZB\succeq 0$ for $B\in\SS^{n}$ and $Z\in\SS^{n}_+$}\nonumber\\
	%&\succ 0. \tag{$\because \theta < \lambda_{\min}(F)$ by Claim~\ref{cl0-mnmxI}}\nonumber
	%\label{eq3--mnmxI}
	\end{align}
	%	It follows that 
	%	\begin{align*}
	%         F'-\theta^*I&=F'-\theta I+(\theta-\theta^*)I\\
	%         &\succeq\frac{1}{2}(F-\theta I)+(\theta-\theta^*)I\\
	%         &\succeq\frac{1}{2}(F-\theta^* I)-{\delta_s}\theta^*I\tag{$\because (1-{\delta_s})\theta^*\le\theta\le\theta^*$}\\
	%         &=\frac{1}{2}\left(F-(1+2{\delta_s})\theta^*I\right)\\
	%         &\succeq\frac{1}{2}\left(F-\frac{1+2{\delta_s}}{1+\bareps/n}\lambda_{\min}(F)I\right)\tag{by Claim~\ref{cl2-mnmxI}}\\
	%         &\succ\frac{1}{2}\left(F-\lambda_{\min}(F)I\right)\succ 0,\\
	%	\end{align*}
	%	where the last inequality follows by the definition of ${\delta_s}$, as
	%	\begin{align*}
	%	\frac{1+2{\delta_s}}{1+\bareps/n}=\frac{1+n/\bareps+2\bareps}{2+n/\bareps+\bareps/n}<1 \tag{$\because \bareps<\frac{1}{6}$.}
	%	\end{align*}
	The claim follows.
\end{proof}

\begin{claim}\label{cl6--mnmxII}
	if $\nu>\bareps$, then $\Tr(B)\ge \frac{\nu^2}{8}$.
\end{claim}	
\begin{proof}
	By the definition of $B$,
	\begin{align*}   
	\Tr(B)&=\frac{n}{\bareps\theta}\left(\tau \Tr(X^{1/2}( F-\hat F)X^{1/2})-(\theta-\theta^*)\Tr(X)\right)\\
	&\ge\frac{n}{\bareps\theta}\left(\tau (X\bullet F- X\bullet\hat F)-(\theta-\theta^*)\right)\tag{$\because \Tr(X)\le 1$ by Claim~\ref{cl0--mnmxII}}\\
	&>\frac{n}{\bareps\theta}\left(\tau (X\bullet F- X\bullet \hat F)-\delta_s\theta\right)\tag{$\because \theta^*\le\theta<(1+\delta_s)\theta^*$}\\
	&=\frac{\nu^2}{4}-\frac{\bareps^2}{32}\tag{by definition of $\tau$ and ${\delta_s}$}\\
	&>\frac{\nu^2}{4}-\frac{\nu^2}{32}>\frac{\nu^2}{8}.\tag{$\because \bareps<\nu\le 1$}
	\end{align*}
\end{proof}

\begin{claim}\label{cl6-mnmxII}
	if $\nu>\bareps$, then $\Tr(B^2)<\frac{\nu^2}{10}$.
\end{claim}	
\begin{proof}
	Write $Y=\tau X^{1/2} FX^{1/2}$ and $\hat Y=X^{1/2}\big(\tau \hat F+(\theta-\theta^*)I\big)X^{1/2}$ and note that both $\hat Y$ and $Y$ are in $\SS_+^n$. It follows by the definition of $B$ that
	\begin{align*}   
	\Tr(B^2)&=\frac{n^2}{\bareps^2\theta^2}\Tr\left((Y-\hat Y)^2\right)\\
	&=\frac{n^2}{\bareps^2\theta^2}\left(\Tr(Y^2)+\Tr(\hat Y^2)-2\Tr(Y \hat  Y)\right)\\
	&\le\frac{n^2}{\bareps^2\theta^2}\left(\Tr(Y^2)+\Tr(\hat Y^2)+2\Tr(Y \hat Y)\right)\tag{$\because \hat Y,Y\in\SS^n_+$}\\
	&\le \frac{n^2}{\bareps^2\theta^2}\left(\Tr( Y^2)+\Tr(\hat  Y^2)+2\sqrt{\Tr(Y^2)\Tr(\hat Y^2)}\right) \tag{by Cauchy-Schwarz Ineq.}	\\
	&\le\frac{n^2}{\bareps^2\theta^2}\left(\Tr(Y)^2+\Tr(\hat Y)^2+2\Tr(\hat Y )\Tr(Y)\right) \tag{$\because \hat Y,Y\in\SS^n_+$}\\
	&=\frac{n^2}{\bareps^2\theta^2}(\Tr(Y)+\Tr(\hat Y))^2\\
	&=\frac{n^2}{\bareps^2\theta^2}(\tau\Tr(X F)+\tau\Tr(X\hat F)+(\theta-\theta^*)\Tr(X))^2\\
	&\le\frac{n^2}{\bareps^2\theta^2}(\tau X\bullet F+\tau X\bullet\hat F+(\theta-\theta^*))^2\tag{$\because \Tr(X)\le 1$ by Claim~\ref{cl0--mnmxII}}\\
	&< \frac{n^2}{\bareps^2\theta^2}\left(\tau X\bullet  F+\tau X\bullet\hat F+\delta_s\theta\right)^2\tag{$\because \theta^*\le\theta<(1+\delta_s)\theta^*$}\\
	&=\left(\frac{\nu}{4}+\frac{\bareps^2}{32}\right)^2\tag{by definition of $\tau$ and ${\delta_s}$}\\
	&< \left(\frac{\nu}{4}+\frac{\nu^2}{32}\right)^2\tag{$\because \bareps<\nu\le 1$}<\frac{\nu^2}{10}.
	\end{align*}
\end{proof}

Define $\Phi(t):=\Phi(\theta^*(t),F(y(t)))$.

\begin{claim}\label{cl7-mnmxII}
	$\Phi(t+1)-\Phi(t) \leq -\frac{\bareps \nu(t+1)^2}{40n}$.
\end{claim}
\begin{proof} 
	Note that Claim~\ref{cl5--mnmxII} implies that $\theta^*$ is feasible to the problem $\min\{\Phi(\xi,F'):\xi \geq \lambda_{\max}(F')\}$. Thus,
	\begin{align*}
	\Phi(t+1)&=\Phi(\theta^*(t+1),F')
	\leq \ln \theta^* - \frac{\bareps}{n}\ln \det (\theta^* I-F').\\
	\therefore \Phi(t+1)-\Phi(t) &\leq -\frac{\bareps}{n}\left(\ln \det(\theta^* I-F')-\ln \det(\theta^* I-F)\right)\\
	&\leq -\frac{\bareps}{n}\left(\ln \det\big(\theta^*I-F'\big)-\ln \det(\theta I-F)\right)\tag{$\because \theta^*\le\theta$}\\
	&= -\frac{\bareps}{n}\ln \det \left(I+B\right)\tag{by Claim~\ref{cl5---mnmxII}}\\
	&=-\frac{\bareps}{n}\sum_{j=1}^n\ln  \left(1+\lambda_j(B)\right)\\
	&\leq -\frac{\bareps}{n}\sum_{j=1}^n\left( \lambda_j(B)-\lambda_j(B)^2\right)\tag{by Claim~\ref{cl5----mnmxI} and $ \ln(1+z) \geq z-z^2, \forall z\geq -0.5$}\\
	&=-\frac{\bareps}{n}\left(\Tr(B)-\Tr(B^2)\right)\\
	&<-\frac{\bareps}{8n}\nu^2+\frac{\bareps}{10n}\nu^2\tag{by Claims~\ref{cl6--mnmxII} and \ref{cl6-mnmxII}}\\
	&=-\frac{\bareps}{40n}\nu^2.
	\end{align*}
\end{proof}

\begin{claim}\label{cl8-mnmxII}
	For any $t,t'$, $$\Phi(t')-\Phi(t) \geq (1-\bareps)\ln \displaystyle\frac{(1-2\bareps)X\bullet A_{i(t)}}{(1-\bareps)^2X \bullet F}+\ln(1-\bareps).$$
\end{claim}
\begin{proof}
	Write $F=F(y(t))$, $\theta^*:=\theta^*(t)$, $\theta:=\theta(t)$, $X:=X(t)$, $F'=F(y(t'))$, $\theta'^{*}:=\theta^*(t')$. Then
	\begin{align*}
	\Phi(t')-\Phi(t) 	&=\ln \frac{\theta'^*}{\theta^*}-\frac{\bareps}{n}\ln \det \left[(\theta^*I-F)^{-1}(\theta'^*I-F')\right]\\
	&=\ln \frac{\theta'^*}{\theta^*}-\frac{\bareps}{n}\ln \det \left[\left(\frac{\bareps \theta}{n}I-(\theta-\theta^*)X\right)^{-1}X(\theta'^*I-F')\right]\tag{by Claim~\ref{cl5-----mnmxII}}\\
	&=\ln \frac{\theta'^*}{\theta^*}-\frac{\bareps}{n} \left[\ln \det\left(\frac{\bareps \theta}{n}I-(\theta-\theta^*)X\right)^{-1}+\ln \det \big[X(\theta'^*I-F')\big]\right]\\
	&\ge\ln \frac{\theta'^*}{\theta^*}-\frac{\bareps}{n} \left[\ln \left(\frac{\bareps \theta}{n}-\delta_s\theta\right)^{-n}+\ln \det \big[X(\theta'^*I-F')\big]\right]\tag{$\because \Tr(X)\le1$  by Claim~\ref{cl0--mnmxII} and $\theta^*\le\theta\le (1+{\delta_s})\theta^*$}\\
	&\ge\ln \frac{\theta'^*}{\theta^*}-\frac{\bareps}{n} \left[\ln\left(\frac{n}{(1-\bareps)\bareps \theta}\right)^n+\ln\det X(\theta'^*I-F')\right]\tag{by defintion  of ${\delta_s}$}\\
	&=\ln \frac{\theta'^*}{\theta^*}-\bareps\ln \frac{n}{(1-\bareps)\bareps \theta}-\frac{\bareps}{n}\ln \left[\det X(\theta'^*I-F')\right]\\
	&=\ln \frac{\theta'^*}{\theta^*}-\bareps\ln \frac{n}{(1-\bareps)\bareps \theta}-\frac{\bareps}{n}\sum_{j=1}^n\ln \lambda_j\big(X(\theta'^*I-F')\big)\\
	&\geq \ln \frac{\theta'^*}{\theta^*}-\bareps\ln \frac{n}{(1-\bareps)\bareps \theta}-\bareps\ln \left(\frac{1}{n}\sum_{j=1}^n\lambda_j\big(X(\theta'^*I-F')\big)\right) \tag{by the concavity of $\ln(\cdot)$}\\
	&= \ln \frac{\theta'^*}{\theta^*}-\bareps\ln \frac{n}{(1-\bareps)\bareps \theta}-\bareps\ln \left(\frac{1}{n}\Tr(\theta'^*X-XF')\right)\\
	&\ge \ln \frac{\theta'^*}{\theta^*}-\bareps\ln \frac{n}{(1-\bareps)\bareps \theta}-\bareps\ln \left(\frac{\theta'^*-X\bullet F'}{n}\right)\tag{$\because \Tr(X)\le1$ by Claim~\ref{cl0--mnmxII}}\\
	&= \ln \frac{\theta'^*}{\theta^*}-\bareps\ln \frac{1}{(1-\bareps)\bareps \theta}-\bareps\ln \left(\theta'^*-X\bullet F'\right)\\
	&\ge\ln \frac{\theta'^*}{\theta^*}-\bareps\ln \frac{1}{(1-\bareps)\bareps \theta}-\bareps\ln \left(\theta'^*-\min_{y\in\RR_+^m:~\b1^Ty= 1} X\bullet F(y)\right)\tag{$\because \b1^Ty(t')=1$ by Claim~\ref{cl5.-mnmxII}}\\
	&=\ln \frac{\theta'^*}{\theta^*}-\bareps\ln \frac{1}{(1-\bareps)\bareps \theta}-\bareps\ln \left( \theta'^*-X\bullet A_{i(t)}\right)\tag{by defintion of $i(t)$}\\
	&\geq\min_{\xi> X\bullet A_{i(t)}} \left\{\ln \frac{\xi}{\theta^*}-\bareps\ln \frac{1}{(1-\bareps)\bareps \theta}-\bareps\ln \left(\xi-X\bullet A_{i(t)}\right)\right\}\\
	&= (1-\bareps)\ln \displaystyle\frac{X\bullet A_{i(t)}}{(1-\bareps)^2\theta}+\ln\frac{(1-\bareps)\theta}{\theta^*}\tag{$\min(\cdot)$ is achieved at $\xi=\frac{X\bullet A_{i(t)}}{1-\bareps}$}\\
	&\ge (1-\bareps)\ln \displaystyle\frac{X\bullet A_{i(t)}}{(1-\bareps)^2\theta}+\ln(1-\bareps)\tag{$\because\theta\ge\theta^*$}\\
	&\ge (1-\bareps)\ln \displaystyle\frac{(1-2\bareps)X\bullet A_{i(t)}}{(1-\bareps)^2X \bullet F}+\ln(1-\bareps).\tag{by Claim~\ref{cl3-mnmxII}}\\
	\end{align*}
\end{proof}
\iffalse
The following claim (similar to the one in \cite{J04}) states that $\Phi(t)$ can be closely approximated by $\ln\theta^*(t)$.
\begin{claim}\label{cl8-mnmxII}
	$(1+\bareps)\ln\theta^*(t)-\bareps\ln\frac{1}{\bareps}\le\Phi(t) \leq (1+\bareps)\ln\theta^*(t)-\frac{\bareps}{n}\ln\frac{1}{\bareps}$.
\end{claim}
\begin{proof}
	By the definition of $\Phi(t)=\Phi(\theta^*,F)$,
	\begin{align}\label{sube1-II}
	(1+\bareps)\ln\theta^*&=\Phi(t)+\frac{\bareps}{n}\sum_{j=1}^n\ln\left(\frac{\theta^*}{\lambda_j(F)-\theta^*}\right).
	\end{align}
	Using the concavity of the $\ln(\cdot)$ in \raf{sube1},
	\begin{align*}
	(1+\bareps)\ln\theta^*&\le\Phi(t)+\bareps\ln\left(\frac1n\sum_{j=1}^n\frac{\theta^*}{\lambda_j(F)-\theta^*}\right)\\
	&\le\Phi(t)+\bareps\ln\left(\frac{\bareps}n\sum_{j=1}^n\frac{\theta^*}{\lambda_j(F)-\theta^*}\right)-\bareps\ln\bareps\\
	&=\Phi(t)-\bareps\ln\bareps\tag{by defintion of $\theta^*$},
	\end{align*}
	from which the lower bound in the claim follows. For the upper bound, we have by \raf{sube1-II}, 
	\begin{align*}
	(1+\bareps)\ln\theta^*&\ge\Phi(t)+\frac{\bareps}n\ln\left(\frac{\theta^*}{\lambda_{\min}(F)-\theta^*}\right)\\
	&\ge\Phi(t)+\frac{\bareps}n\ln\frac1{\bareps}\tag{by Claim~\ref{cl2-mnmxII}}.
	\end{align*}
\end{proof}
\fi
\begin{claim}\label{cl9-mnmxII} $\frac{X(0) \bullet A_{i(0)}}{X(0) \bullet F(y(0))}\geq \frac 1\psi:=\frac{\lambda_{\min}(A_{i(0)})}{\lambda_{\max}( A_{i'})}\geq\frac{\min_i\lambda_{\min}(A_i)}{\max_i\lambda_{\max}( A_i)}\ge \frac{\epsilon}{n^2 2^{2\cL}}$, where $i'$ is the index such that $y(0)=\b1_{i'}$ .\end{claim}
\begin{proof}
	Let $X(0)=\sum_{j=1}^n\lambda_ju_ju_j^T$ be the spectral decomposition of $X(0)$. Then,  
	\begin{align*}
	X(0)\bullet A_{i(0)}&=\sum_{j=1}^n\lambda_jA_{i(0)}\bullet u_ju_j^T	\geq\sum_{j=1}^n\lambda_j\lambda_{\min}(A_{i(0)})=\lambda_{\min}(A_{i(0)})\cdot\Tr(X(0))\\
	X(0)\bullet F(y(0))&=\sum_{j=1}^n\lambda_jA_{i'}\bullet u_ju_j^T	\leq\sum_{j=1}^n\lambda_j\lambda_{\max}(A_{i'})=\lambda_{\max}( A_{i'})\cdot\Tr(X(0)). 
	\end{align*}
	Note that $\psi\le\frac{n^2 2^{2\cL}}{\epsilon}$ by Assumption~(B-II). 
	The claim follows.
\end{proof}

\begin{claim}\label{cl10-mnmxII}
	The algorithm terminates in at most $O(n\log \psi+\frac{n}{\epsilon^2})$ iterations. %, where $\psi\ge \frac{\epsilon^2}{6n^2}$.
\end{claim}

\begin{proof}
	Let $t_{-1}=-1$ and, for $s=0,1,2,\ldots$, let $t_s$ be the smallest $t$ such that $\nu(t+1)\le2^{-(s+1)}$ (so $t_s+1$ is the value of $t$ at which the iteration $s+1$ of the outer while-loop starts). Then for $t=t_{s-1}+1,\ldots,t_{s}-1$, we have $\nu(t+1)>2^{-(s+1)}=2\eps_s$. Hence, for $s=0$,
	\begin{align*}
	-\frac{\eps_0^3 t_0}{40n}&>\Phi(t_0)-\Phi(0)	\tag{by Claim~\ref{cl7-mnmxII}} \\ &\ge(1-\eps_0)\ln \displaystyle\frac{(1-2\eps_0)X(0)\bullet A_{i(0)}}{(1-\eps_0)^2X(0) \bullet F(y(0))}+\ln(1-\eps_0)	\tag{by Claim~\ref{cl8-mnmxII}} \\ 
	&\ge(1-\eps_0)\ln \frac{\psi(1-2\eps_0)}{(1-\eps_0)^2}+\ln(1-\eps_0)	\tag{by Claim~\ref{cl9-mnmxII}}.
	\end{align*}
	Setting $\eps_0=\frac{1}{4}$ in the last series of inequalities we get 
	\begin{align}\label{est1-II}
	t_0&<1920n\ln\big(\frac{9\psi}{8}\big)+\ln\frac43=O(n\log \psi).
	\end{align}
	Now consider $s\ge 1$: 
	\begin{align*}
	-\frac{\eps_s^3 (t_s-t_{s-1})}{40n}&>\Phi(t_s)-\Phi(t_{s-1})	\tag{by Claim~\ref{cl7-mnmxII}} \\ &\ge(1-\eps_s)\ln\frac{(1-2\eps_s)X(t_{s-1})\bullet A_{i(t_{s-1})}}{(1-\eps_s)^2X(t_{s-1})\bullet F(y(t_{s-1}))}+\ln(1-\eps_s)	\tag{by Claim~\ref{cl8-mnmxII}} \\ 
	&=
	(1-\eps_s)\ln \frac{(1-2\eps_s)\big(1-\nu(t_{s-1}+1)\big)}{(1-\eps_s)^2\big(1+\nu(t_{s-1}+1)\big)}+\ln(1-\eps_s)	\tag{by definition of ~$\nu(t_{s-1}+1)$}\\
	&\ge(1-\eps_s)\ln \frac{(1-2\eps_s)(1-4\eps_s)}{(1-\eps_s)^2(1+4\eps_s)}+\ln(1-\eps_s)	\tag{$\because \nu(t_{s-1}+1)\le2^{-s}=4\eps_s$}\\
	&\ge-(1-\eps_s)\ln(1+32\eps_s)-\ln(1+3\eps_s)	\tag{$\because\eps_s\le\frac18$}
	>-35\eps_s.
	\end{align*}
	Setting $\eps_s=\frac{1}{2^{s+2}}$ in the last series of inequalities we get 
	\begin{align}\label{est2-II}
	t_s-t_{s-1}&<\frac{1400n}{\eps_s^2}=O(n/\eps_s^2).
	\end{align}
	Summing \raf{est1-II}, and \raf{est2-II} over $s=1,2,\ldots,\lceil\log\frac{1}{\epsilon}\rceil$, we get the claim. %Finally, note that $\psi\ge \frac{\epsilon^2}{6n^2}$ by assumption (B-II).
\end{proof}

\subsubsection{Primal Dual Feasibility and Approximate Optimality  }
Let $t_f+1$ be the value of $t$ when the algorithm terminates and $s_f+1$ be the value of $s$ at termination. For simplicity, we write  $s=s_f$.

\begin{claim}\label{cl11-mnmxII}(Primal feasibility).
	$\hat X\succ 0$ and $\min_i  A_i\bullet \hat X\geq 1$.
\end{claim}
\begin{proof}
	The first claim is immediate from Claim~\ref{cl0-mnmxII}. To see the second claim, we use the definition of $\nu(t_f)$ and the termination condition in line~\ref{s1.-mnmxI} (which is also satisfied even if $X(t_f)\bullet  F(y(t_f))-X(t_f)\bullet A_{i(t_f)} =0$):
	\begin{align*}
	\frac{X(t_f)\bullet  F(y(t_f))-X(t_f)\bullet A_{i(t_f)}}{X(t_1)\bullet A_{i(t_f)} + X(t_f)\bullet  F(y(t_f))} &\leq \bareps.\\
	\therefore (1-\bareps)X(t_f)\bullet F(y(t_f)) &\leq (1+\bareps)X(t_f)\bullet A_{i(t_f)}\\
	&= (1-\bareps)\min_i X(t_f)\bullet A_i \tag{by the defintition of $i(t_f)$}\\
	\therefore (1-\bareps)(1-2\bareps) \theta(t_f) &< (1+\bareps)\min_i X(t_f)\bullet A_i \tag{$\because X(t_f) \bullet F(y(t_f)) > (1-2\bareps)\theta(t_f)$ by Claim~\ref{cl3-mnmxII}}.
	\end{align*}
	The claim follows by the definition of $\hat X$ in step~\ref{so-mnmxI} of the algorithm.
\end{proof}
\begin{claim}(Dual feasibility).\label{cl12-mnmxII}
	$\hat y\ge0$ and $F(\hat y)\prec I$.
\end{claim}
\begin{proof}
	The fact that $\hat y\ge0$ follows from the initialization of $y(0)$ in step~\ref{s1-mnmxI}, Claim~\ref{cl4-mnmxI},  and the update of $y(t+1)$ in step~\ref{s6-mnmxI}. For the other claim, we have \begin{align*}
	\lambda_{\max}\big(F(\hat y)\big)&=\frac{1}{\theta(t_f)}\lambda_{\max}\big(F(y(t_f))\big)
	\le1-\frac{\bareps}n.\tag{by Claim~\ref{cl2-mnmxII}}
	\end{align*} 
\end{proof}

\begin{claim}(Approximate optimality).\label{cl13-mnmxII}
	$I\bullet \hat X\le\frac{1+\bareps}{(1-2\bareps)^2}\b1^T\hat y$.
\end{claim}
\begin{proof}
	By Claim~\ref{cl0--mnmxII}, we have $\Tr(X(t_f)\le1$, and by Claim~\ref{cl5.-mnmxII}, we have $\b1^Ty(t_f)=1$. The claim follows by the definition of $\hat X$ and $\hat y$ in step~\ref{so-mnmxII}.   
\end{proof}	

\begin{remark}\label{r3}
	Similar to the packing case, suppose that in step~\ref{s3-mnmxII} of Algorithm~\ref{log-cover-alg}, we instead define $i(t)$ to be an index $i\in[m]$ such that $A_i\bullet X(t)\le 1+\eps_s$, and we are guaranteed that such index exists in each iteration of the algorithm. Then the dual solution $\hat y$ satisfies: $\b1^T\hat y\ge 1-O(\epsilon)$. Indeed, the proof of Claim~\ref{cl11-mnmxII} can be easily modified to show that $\theta(t_f)\le\frac{(1+\eps_{s_f})^2}{(1-2\eps_{s_f})^2}$, which combined with the definition of $\hat y$ in step~\ref{so-mnmxI} of the algorithm implies the claim.
\end{remark}	

\subsubsection{Running Time per Iteration}

\paragraph{Computing $\theta(t)$.} Given $F:=F(y(t))\succeq 0$, we first compute an approximation $\widetilde\lambda$ of $\lambda_{\max}(F)$ using Lanczos' algorithm with a random start.
By Claim~\ref{cl2-mnmxII}, we need $\widetilde\lambda$ to lie in the range $[\lambda_{\max}(F),\frac{\lambda_{\max}(F)}{1-\eps_s/n}]$.
To obtain $\widetilde\lambda$, we  apply Lemma~\ref{KW} with $M:=F^{n}$ and $\gamma:=\frac{\eps_s}{2}$. Then in $O\big(\frac{\log n}{\sqrt{\eps_s}}\big)$ iterations we get $\widetilde \lambda:=\left(\frac{v^TF^{n}v}{1-\gamma}\right)^{1/n}$(where $v$ be the vector obtained from Lemma~\ref{KW}) satisfying
\begin{align}\label{ee--II}
\lambda_{\max}(F)\le\widetilde\lambda \le\frac{\lambda_{\max}(F)}{(1-\gamma)^{1/n}}\le (1+\eps_s)^{1/n}\lambda_{\max}(F)\le\frac{\lambda_{\max}(F)}{1-\eps_s/n}.
\end{align}

%Note that we can compute $F^{-n}$in $O(n^{\omega}\log n)$, where $w$ is the {\it exponent of matrix multiplication}. 
Thus, the overall running time for computing $\widetilde\lambda$ is $O(n^{\omega}\log n+\frac{n^2\log n}{\sqrt{\eps_s}})$.
Given $\widetilde\lambda$, we know by Claim~\ref{cl2-mnmxII} and \raf{ee--II} that $\theta^*(t)\in[\widetilde\lambda,\frac{\widetilde\lambda}{1-\eps_s}]$. Then we can apply binary search to find $\theta(t):=\theta^*(t)^{\delta_s}$ as follows. Let $\theta_k=\widetilde\lambda(1+\delta_s)^k$, for $k=0,1,\ldots,K:=\lceil\frac{-2\ln(1-\eps_s)}{\delta_s}\rceil$, and note that $\theta_K\ge\widetilde\lambda$. Then we do binary search on the exponent $k\in\{0,1\ldots,K\}$; each step of the search evaluates $g(\theta_k):=\displaystyle\frac{\bareps \theta_k}{n}\Tr(\theta_k I-F)^{-1}$, and depending on whether this value is less than or at least $1$, the value of $k$ is decreased or increased, respectively. The search stops when the search interval $[\ell,u]$ has $u\le \ell+1$, in which case we set $\theta(t)=\theta_u$; the number of steps until this happens is $O(\log K)=O(\log\frac{1}{\delta_s})=O(\log\frac n{\eps_s})$. By the monotonicity of $g(x)$ (in the interval $[\lambda_{\max}(F),+\infty)$), and the property of binary search, we know that $\theta^*\in[\theta_\ell,\theta_u]$. Thus, by the stopping criterion, $$\theta_\ell\le\theta^*(t)\le\theta(t)=\theta_u\le\theta_{\ell+1}=(1+\delta_s)\theta_\ell,$$
implying that $\theta^*(t)\le\theta(t)\le(1+\delta_s)\theta^*(t)$.
Since evaluating $g(\theta_k)$ takes $O(n^\omega)$, the overall running time for the binary search procedure is $O(n^{\omega}\log\frac n{\eps_s})$, and hence the total time needed for for computing $\theta(t)$ is $O(n^{\omega}\log \frac n{\epsilon}+\frac{n^2\log n}{\sqrt{\epsilon}})$.

As all other steps of the algorithm inside the inner while-loop can be done in $O(\cT+n^2)$ time, where $\cT$ is the time taken by a single call to the minimization oracle in step~\ref{s3-mnmxII}, in view of Claim~\ref{cl10-mnmxI}, we obtain the following  result.
\begin{theorem}\label{thm:cover-main}
	For any $\epsilon>0$, Algorithm~\ref{log-cover-alg} outputs an $O(n\log \psi+\frac{n}{\epsilon^2})$-sparse $O(\epsilon)$-optimal primal-dual pair in time $O((n\log \psi+\frac{n}{\epsilon^2})(n^{\omega}\log \frac n{\epsilon}+\frac{n^2\log n}{\sqrt{\epsilon}}+\cT))=\tilde O(\frac{n^{\omega+1}\log \psi}{\epsilon^{2.5}}+\frac{n\cT \log n}{\epsilon^2})$. 
\end{theorem}	

\section{Applications}\label{app}

\subsection{Robust Packing and Covering SDPs}

Consider a packing-covering pair of the form~\raf{packI}-\raf{coverI} or ~\raf{packII}-\raf{coverII}. 
In the framework of {\it robust optimization} (see, e.g. \cite{BEN09,BN02}), we assume that each constraint matrix $A_i$ is not known exactly; instead, it is given by a convex uncertainty set $\cA_i\subseteq\SS^{n}_+$.
It is required to find a (near)-optimal solution for the packing-covering pair under the {\it worst-case} choice $A_i\in\cA_i$ of the constraints in each uncertainty set. A typical example of a {\it convex} uncertainty set is given by an {\it affine perturbation} around a nominal matrix $A_i^0\in\SS_+^n$: 
\begin{align}
\cA_i=\left\{A_i:=A_i^0+\sum_{r=1}^k\delta_rA_i^r:~\delta=(\delta_1,\ldots,\delta_k)\in\cD\right\}, 
\end{align}
where  $A_i^1,\ldots,A_i^k\in\SS^n_+$, and $\cD\subseteq\RR^k_+$ can take, for example, one of the following forms: 
\begin{itemize}
	\item {\it Ellipsoidal uncertainty}: $\cD=E(\delta_0,D):=\{\delta\in\RR^{k}_+:(\delta-\delta_0)^TD^{-1}(\delta-\delta_0)\le 1\}$, for given positive definite matrix $D\in\SS_+^k$ and vector $\delta_0\in\RR_+^k$ such that $E(\delta_0,D)\subseteq\RR_+^k$;
	\item {\it Box uncertainty}: $\cD=B(\delta_0,\rho):=\{\delta\in\RR^{k}_+:\|\delta-\delta_0\|_1\le\rho\}$, for given positive number $\rho\in\RR_+$ and vector $\delta_0\in\RR_+^k$ such that $B(\delta_0,\rho)\subseteq\RR_+^k$;
	\item {\it Polyhedral uncertainty}: $\cD:=\{\delta\in\RR^{k}_+:D\delta\le w\}$, for given matrix $D\in\RR^{h\times k}$ and vector $w\in\RR^{h}$.
\end{itemize}
Without loss of generality, we consider the robust version of  ~\raf{npackI}-\raf{ncoverI},  where $A_i$, for $i\in[m]$, belongs to a convex uncertainty set $\cA_i$. Then the robust optimization problem and its dual can be written as follows:

{\centering \hspace*{-18pt}
	\begin{minipage}[t]{.47\textwidth}
		\begin{alignat}{3}
		\label{rpackI}
		\tag{\rP-I} \quad& \displaystyle z_P^* = \max\quad I\bullet X\\
		\text{s.t.}\quad & \displaystyle A_i\bullet X\leq 1, \quad \forall A_i\in\cA_i\quad \forall i\in [m]\nonumber\\
		\qquad &X\in\RR^{n\times n},~X\succeq  0\nonumber
		\end{alignat}
	\end{minipage}
	\,\,\, \rule[-18ex]{1pt}{18ex}
	\begin{minipage}[t]{0.47\textwidth}
		\begin{alignat}{3}
		\label{rcoverI}
		\tag{\rC-I} \quad& \displaystyle z_D^* = \inf\quad \sum_{i=1}^m\int_{\cA_i}y_{A_i}^idA_i\\
		\text{s.t.}\quad & \displaystyle \sum_{i=1}^m\int_{\cA_i}y_{A_i}^iA_idA_i\succeq I\nonumber\\
		\qquad &y^i \text{ is a {\it discrete measure} on }\cA_i, \quad  \forall i\in [m].\nonumber
		\end{alignat}
\end{minipage}}

%\vspace{0.2in}

\noindent As before, we assume~(B-I), where $A_{1},\ldots,A_{r}\in\bigcup_{i\in[m]}\cA_i$.  We call  a pair of solutions $(X,y)$ to be $\epsilon$-optimal for~\raf{rpackI}-\raf{rcoverI}, if 
$$z_P^*\ge I\bullet X\geq(1-\epsilon)\sum_{i=1}^m\int_{\cA_i}y_{A_i}^idA_i\ge(1-\epsilon) z_D^*.$$
As a corollary of Theorem~\ref{thm:pack-main}, we obtain the following result.
\begin{theorem}\label{thm:pack-roust}
	For any $\epsilon>0$, Algorithm~\ref{log-pack-alg} outputs an $O(\epsilon)$-optimal primal-dual pair for~\raf{rpackI}-\raf{rcoverI} in time $\tilde O\big(\frac{n^{\omega+1}\log\psi}{\epsilon^{2.5}}+\frac{n\cT\log\psi}{\epsilon^2}\big)$, where
	$\psi:=\frac{r\cdot\max_{i\in[m],A_i\in\cA_i}\lambda_{\max}(A_i)}{\lambda_{\min}(\bar A)}$ and $\cT$ is the time to compute, for a given $Y\in\SS_+^n$, a pair $(i,A_i)$ such that
	\begin{align}\label{oracle-}
	\displaystyle (i,A_i)\in\argmax_{i\in[m],~A_i\in\cA_i}A_i\bullet Y.
	\end{align}
\end{theorem}	
	Note that \raf{oracle-} amounts to solving a linear optimization problem over a convex set. Moreover, for simple uncertainty sets, such as balls or ellipsoids, such computation can be done very efficiently. 
\subsection{Carr-Vempala Type Decomposition}

Consider a maximization (resp., minimization) problem over a discrete set $\cS\subseteq\ZZ^n$ and a corresponding SDP-relaxation over $\cQ\subseteq\SS^n_+$:

{\centering \hspace*{-18pt}
	\begin{minipage}[t]{.47\textwidth}
		\begin{alignat}{3}
		\label{COP}
		\tag{COP} \quad& \displaystyle z_{CO}^* = \left\{\begin{array}{c}\max \\ \min\end{array}\right\}\quad C\bullet qq^T\\
		\text{s.t.}\quad & \displaystyle q\in \cS\nonumber
		\end{alignat}
	\end{minipage}
	\,\,\, \rule[-14ex]{1pt}{14ex}
	\begin{minipage}[t]{0.47\textwidth}
		\begin{alignat}{3}
		\label{sdprlx}
		\tag{SDP-RLX} \quad& \displaystyle z_{SDP}^* = \left\{\begin{array}{c}\max \\ \min\end{array}\right\}\quad C\bullet Q\\
		\text{s.t.}\quad & \displaystyle Q\in\cQ\nonumber,
		\end{alignat}
\end{minipage}}

\noindent where $C\in\SS_+^n$. 

\begin{definition}\label{d1}
	For $\alpha\in(0,1]$ (resp., $\alpha\ge 1$), an $\alpha$-{\it integrality gap verifier} $\cA$ for \raf{sdprlx} is a polytime algorithm that, given any $C\in\SS_+^n$ and any $Q\in\cQ$ returns a $q\in\cS$ such that $B\bullet qq^T\ge\alpha B\bullet Q$ (resp., $C\bullet qq^T\le\alpha C\bullet Q$). %An integrality gap verifier $\cA$ is said to be {\it oblivious} (see, e.g., \cite{FFT16}) if $\cA(C,Q)=\cA(Q)$ does not depend on the objective $B$. 
\end{definition}
For instance, if $\cS=\{-1,1\}^n$ and $Q=\{X\in\SS_+^n:X_{ii}=1~~\forall i\in[n]\}$, then a $\frac2\pi$-integrality gap verifier for the maximization version of \raf{sdprlx} is known \cite{N97}.

Carr and Vempala \cite{CV02} gave a decomposition theorem that allows one to use an $\alpha$-integrality gap verifier for a given LP-relaxation of a combinatorial maximization (resp., minimization) problem, to decompose a given fractional solution to the LP into a convex combination of integer solutions that is dominated by (resp., dominates) $\alpha$ times the fractional solution. 
In \cite{EMO18}, we prove a similar result for SDP relaxations:

\begin{theorem}\label{CV-SDP}
	Consider a combinatorial  maximization (resp., minimization) problem~\raf{COP} and its SDP relaxation~\raf{sdprlx}, admitting an $\alpha$-integrality gap verifier $\cA$. Assume the set $\cS$ is full-dimensional. Then there is a polytime algorithm that, for any given $Q\in\cQ$, finds a set $\cX\subseteq\cS$, of {\it polynomial} size, and a set of convex multipliers $\{\lambda_q\in\RR_+:~q\in\cX\}$, $\sum_{q\in\cX}\lambda_q=1$,  such that 
	\begin{align*}%\label{cvx-comb}
	\alpha Q\preceq\sum_{q\in\cX}\lambda_qqq^T~~~\quad\text{(resp.,  }\alpha Q\succeq\sum_{q\in\cX}\lambda_qqq^T\text{)}.
	\end{align*}
\end{theorem}    

The proof of Theorem~\ref{CV-SDP} is obtained by considering the following pairs of packing and covering SDPs (of types I and II, respectively): 

{\centering \hspace*{-18pt}
	\begin{minipage}[t]{0.47\textwidth}
		\begin{alignat}{3}
		\label{cvx}
		\tag{CVX-I} \quad& \displaystyle z_I^* = \min\quad \sum_{q\in\cS}\lambda_q\\
		\text{s.t.}\quad & \displaystyle \sum_{q\in\cS}\lambda_qqq^T\succeq \alpha Q\label{ecv0}\\
		& \displaystyle\sum_{q\in\cS}\lambda_q\geq1\label{ecv1}\\
		\qquad &\lambda\in\RR^{\cS},~\lambda\geq 0\nonumber
		\end{alignat}
	\end{minipage}
	\,\,\, \rule[-20ex]{1pt}{20ex}
	\begin{minipage}[t]{.47\textwidth}
		\begin{alignat}{3}
		\label{cvx-dual}
		\tag{CVX-dual-I} \quad& \displaystyle z_I^* = \max\quad \alpha Q\bullet Y+u\\
		\text{s.t.}\quad & \displaystyle qq^T\bullet Y+u\leq 1, \forall q\in \cS\label{ecv2}\\
		\qquad &Y\in\SS^{n}_+,~u\ge 0\nonumber.
		\end{alignat}
\end{minipage}}

{\centering \hspace*{-18pt}
	\begin{minipage}[t]{0.47\textwidth}
		\begin{alignat}{3}
		\label{cvx-}
		\tag{CVX-II} \quad& \displaystyle z_{II}^* = \max\quad \sum_{q\in\cS}\lambda_q\\
		\text{s.t.}\quad & \displaystyle \sum_{q\in\cS}\lambda_qqq^T\preceq \alpha Q\label{ecv0-}\\
		& \displaystyle\sum_{x\in\cS}\lambda_q\leq1\label{ecv1-}\\
		\qquad &\lambda\in\RR^{\cS},~\lambda\geq 0\nonumber
		\end{alignat}
	\end{minipage}
	\,\,\, \rule[-20ex]{1pt}{20ex}
	\begin{minipage}[t]{.47\textwidth}
		\begin{alignat}{3}
		\label{cvx-dual-}
		\tag{CVX-dual-II} \quad& \displaystyle z_{II}^* = \min\quad \alpha Q\bullet Y+u\\
		\text{s.t.}\quad & \displaystyle qq^T\bullet Y+u\geq 1, \forall q\in \cS\label{ecv2-}\\
		\qquad &Y\in\SS_+^{n},~u\ge 0\nonumber.
		\end{alignat}
\end{minipage}}

\noindent It can be shown, using the fact that the SDP relaxation admits an $\alpha$-integrality gap verifier, that $z^*_I=z^*_{II}=1$, and that the two primal-dual pairs can be solved in polynomial time using the Ellipsoid method. 
Here, we derive a more efficient but approximate version of Theorem~\ref{CV-SDP}.	

\begin{theorem}\label{CV-SDP-approx}
	Consider a combinatorial  maximization (resp., minimization) problem~\raf{COP} and its SDP relaxation~\raf{sdprlx}, admitting an $\alpha$-integrality gap verifier $\cA$. Assume the set $\cS$ is full-dimensional and let $\epsilon>0$ be a given constant. Then there is a polytime algorithm that, for any given $Q\in\cQ$, finds a set $\cX\subseteq\cS$ of  size $|\cX|= O(\frac{n^3}{\epsilon^2}\log (nW))$ (resp., of size $|\cX|=O(n\log \frac{n}{\epsilon}+\frac{n}{\epsilon^2})$), where  $W:=\max_{q\in\cS,~i\in[n]}|q_i|$,  and a set of convex multipliers $\{\lambda_q\in\RR_+:~q\in\cX\}$, $\sum_{q\in\cX}\lambda_q=1$,  such that 
	\begin{align}\label{cvx-comb-}
	(1-O(\epsilon))\alpha Q\preceq\sum_{q\in\cX}\lambda_qqq^T~~~\quad\text{(resp.,  }(1+O(\epsilon))\alpha Q\succeq\sum_{q\in\cX}\lambda_qqq^T\text{)}.
	\end{align}
\end{theorem}
\begin{proof}
	Let us first consider the maximization problem and the corresponding covering SDP~\raf{cvx}. We can write \raf{cvx}-\raf{cvx-dual} in the form of \raf{coverI}-\raf{packI}, where the set of constraints $[m]$ corresponds to $\cS$, by setting \begin{align}\label{eeeee0}
	A_q&:=\left[\begin{array}{ll}qq^T & 0\\ 0 &1\end{array}\right],\qquad C:=\left[\begin{array}{ll}\alpha Q&0\\0&1\end{array}\right], \qquad X:=\left[\begin{array}{ll}Y&0\\0&u\end{array}\right].
	\end{align}
	Let us fix any linearly independent subset $\cS'\subseteq\cS$ of $\cS$ of size $n$. Write $\bar A:=\sum_{q\in\cS'}qq^T$. Then for any  $Y\succeq0$, feasible for \raf{cvx-dual}, we have
	$
	I\bullet Y\le \frac{\bar A\bullet Y}{\lambda_{\min}(B)}\le\frac{n}{\lambda_{\min}(\bar A)}.
	$
	To arrive at a bound $\tau$ as in Assumption~(B-I), we need to lower-bound $\lambda_{\min}(\bar A)$. 
	Let $\cL'$ be the total bit length needed to describe $\cS'$. Then we have the following bound.
	\begin{claim}\label{cl-cvx-I}
		$\lambda_{\min}(\bar A)\ge \gamma:=2^{-2\cL'-1}$.
	\end{claim}
	\begin{proof}
		Equivalently, we need to show that $\sum_{q\in\cS'} (q^Tv)^2+v_0^2\ge \gamma$, for any unit vector $(v,v_0)\in\RR^{n+1}$. Suppose for the sake of contradiction that $|v_0|<\sqrt{\gamma}$ and $|q^Tv|<\sqrt{\gamma}$ for all $q\in \cS$. Let $H\in\RR^{n\times n}$ be the matrix whose columns are the vectors $q\in\cS'$ and $h\in\RR^n$ be a  vector with component $q^Tv$ in the position corresponding to $q\in\cS'$. Then the linear system $Hx=h$ has a unique solution $x=v=H^{-1}h$ such that each component is bounded in absolute value from above by $2^{\cL'}\sqrt{\gamma}$ (see. e.g., \cite[chapter 1]{GLS88}).  Since $(v,v_0)$ is a unit vector, it follows that 
		$$
		1=\|v\|^2+v_0^2<2^{2\cL'}\gamma+\gamma <1,
		$$
		a contradiction.
	\end{proof}
	From Claim~\ref{cl-cvx-I}, we know that assumption~(B-I) is satisfied with $\tau:=2^{2\cL'+1}n+1$, where $\cL'\le n^2 \log (W+1)$.
	Let $\alpha Q=L^TD L$ be the LDL-decomposition of $\alpha Q$ and write $U:=L^{-1}$. 
	By the reduction in Appendix~\ref{normal-I}, 
	we can use $\alpha Q(\delta)=L^TD(\delta)L=\alpha Q+\delta L^T\bar I L$, where $D(\delta)=D+\delta L^T\bar I L$ and $\delta\le\frac{\epsilon}{\tau I\bullet L^TL}$ (as $z_I^*=1$), instead of $\alpha Q$ without changing the objective value by a factor more than $(1+\epsilon)$ (if $Q$ is nonsingular, then we set $\delta=0$). (Recall that $\bar I$ is the $0/1$-diagonal matrix with ones only in the entries corresponding to the {\it zero} diagonal entries of the diagonal matrix $D$, and note that the matrix $L$ is {\it independent} of $\delta$.) For $q\in\cS$, let $p(q):=D(\delta)^{-1/2}U^Tq$. Using the transformation of variables $Y':= D(\delta)^{1/2}LYL^TD(\delta)^{1/2}$, we get $\alpha Q(\delta)\bullet Y=I\bullet Y'$ and $qq^T\bullet Y=p(q)p(q)^T\bullet Y'.$ Hence, we obtain a normalized form of (an approximate version of) \raf{cvx}-\raf{cvx-dual}, where $q\in\cS$ is replaced by $p(q)$. In view of Remark~\ref{r2}, it is enough to show that in each iteration $t$ of Algorithm~\ref{log-pack-alg}, we can find efficiently a $q\in\cS$ such that $p(q)p(q)^T\bullet Y'+u\ge 1-O(\eps_s)$ for given $Y'=Y'(t)\succeq 0$ and $u=u(t)\ge 0$ such that $\Tr(Y')+u\in (1-\eps_s,1)$ (by Claim~\ref{cl0--mnmxI}, where $X(t):=\left(\begin{array}{ll}Y'&0\\0&u\end{array}\right)$ in step~\ref{s3.-mnmxI} of the algorithm). To do this, let $Y:=UD(\delta)^{-1/2}Y'D(\delta)^{-1/2}U^T$ and call the integrality gap verifier $\cA$ on $(Y,Q)$ to get a vector $q\in\cS$ such that $qq^T\bullet Y\ge\alpha Q\bullet Y$. Then
	\begin{align}\label{e--1}
	p(q)p(q)^T\bullet Y'+u&=qq^T\bullet Y+u\ge\alpha Q\bullet Y+u=\alpha Q(\delta)\bullet Y+u-\delta L^T\bar IL\bullet Y=I\bullet Y'+u-\delta L^T\bar IL\bullet Y.
	\end{align}\label{e--2}
	We bound the ``error term" $\delta L^T\bar IL\bullet Y$ using the definition of $Y'=Y'(t):=\frac{\bareps\theta(t)}{n+1}\big(\sum_{q\in\cS}\lambda_q(t)p(q)p(q)^T-\theta(t) I\big)^{-1}$ in step~\ref{s3.-mnmxI} of the algorithm as follows:
	\begin{align}
	\delta L^T\bar IL\bullet Y&=\delta L^T\bar IL \bullet UD(\delta)^{-1/2}Y'D(\delta)^{-1/2}U^T=\delta \bar I\bullet D(\delta)^{-1/2} Y'D(\delta)^{-1/2}=\delta D(\delta)^{-1/2} \bar ID(\delta)^{-1/2}\bullet Y'=\bar I\bullet Y'\nonumber \\
	&=\frac{\bareps\theta(t)}{n+1}\bar I\bullet\big(\sum_{q\in\cS}\lambda_q(t)p(q)p(q)^T-\theta(t) I\big)^{-1}=\frac{\bareps\theta(t)}{n+1}\bar I\bullet\big(D(\delta)^{-1/2}H(t)D(\delta)^{-1/2}-\theta(t)I\big)^{-1},\label{eee2}
	\end{align}
	where, for brevity, we write $H=H(t):=\sum_{q\in\cS}\lambda_q(t)U^Tqq^TU$. To bound \raf{eee2}, we need to compute the submatrix of $\big(D(\delta)^{-1/2}H(t)D(\delta)^{-1/2}-\theta(t)I\big)^{-1}$ corresponding to the non-zeros of $\bar I$. Let the corresponding decompositions of the matrices $D(\delta)$ and $G(t):=D(\delta)^{-1/2}H(t)D(\delta)^{-1/2}-\theta(t)I$ be as follows:
	\begin{align}\label{eee0}
	D(\delta)=\left(\begin{array}{ll}D'&0\\0&\delta\bar I\end{array}\right), \qquad H(t)=\left(\begin{array}{ll}H_1&H_2\\H_2^T&H_3\end{array}\right), \qquad G(t)=\left(\begin{array}{ll}G_1&G_2\\G_2^T&G_3\end{array}\right)=\left(\begin{array}{ll}(D')^{-1/2}H_1(D')^{-1/2}-\theta(t)I&\frac{1}{\sqrt{\delta}}(D')^{-1/2}H_2\\ \frac{1}{\sqrt{\delta}} H_2^T(D')^{-1/2}&\frac1\delta H_3-\theta(t)I\end{array}\right),
	\end{align}
	where, for simplicity, we use $I$ to denote the identity matrix of the proper dimension, according to the context. As $G(t)\succ 0$, we have 
	\begin{align}\label{eeee1-}
	\theta(t)\le\lambda_{\min}\big((D')^{-1/2}H_1(D')^{-1/2}\big), \text{ and } M:=H_3-H_2^T(D')^{-1}H_1(D')^{-1}H_2\succ 0.
	\end{align}
	Using the block inversion formula: 
	\begin{align}
	\bar I\bullet G(t)^{-1}&=I\bullet\big(G_3-G_2^TG_1G_2\big)^{-1}= I\bullet\Big(\frac1\delta H_3-\theta(t)I-\frac1\delta H_2^T(D')^{-1/2}\big((D')^{-1/2}H_1(D')^{-1/2}-\theta(t)I\big)(D')^{-1/2}H_2 \Big)^{-1}\nonumber\\ \label{eee3}
	&=\delta I\bullet\Big(H_3-\delta\theta(t)I- H_2^T(D')^{-1/2}\big((D')^{-1/2}H_1(D')^{-1/2}-\theta(t)I\big)(D')^{-1/2}H_2\Big)^{-1},
	\end{align}
and writing $\bar M:=H_3- H_2^T(D')^{-1/2}\big((D')^{-1/2}H_1(D')^{-1/2}-\theta(t)I\big)(D')^{-1/2}H_2$, we get by~\raf{eee3},
	\begin{align}\label{eeee3-}
	\bar I\bullet G(t)^{-1}&=\sum_{j=1}^n\frac{\delta}{\lambda_j(\bar M)-\delta\theta(t)} \le \frac{\delta n}{\lambda_{\min}(\bar M)-\delta\theta(t)}.
	\end{align}
	Note that $\bar M=M+\theta(t)H_2^T(D')^{-1}H_1(D')^{-1}H_2\succeq M\succ 0$ by \raf{eeee1-}, and that $M$, $D'$ and $H_1$ are independent of $\delta$. It follows that, if we set 
	\begin{align}\label{eeee2-}
	\delta:=\min\left\{\frac{\epsilon}{\tau I\bullet L^TL},\frac{\lambda_{\min}(M)}{2\lambda_{\min}((D')^{-1/2}H_1(D')^{-1/2})}\right\},
	\end{align}
	then by \raf{eee2}, \raf{eeee1-} and \raf{eeee3-}, 
	\begin{align}\label{eeee4-}
	\delta L^T\bar IL\bullet Y&\le\frac{\bareps\theta(t)}{n+1}\bar I\bullet G(t)^{-1}\le \frac{\bareps\theta(t)}{n+1}\cdot\frac{\delta n}{\lambda_{\min}(M)-\delta\theta(t)}\le\frac{\bareps n}{n+1}<\bareps.
	\end{align}
	Using \raf{eee2}, \raf{eeee4-} and $I\bullet Y'+u\ge 1-\bareps$, we get the desired inequality. Let $\cX\subseteq\cS$ be the set of vectors $q\in\cS'$ such that $\lambda_q>0$ when the algorithm terminates. Since each iteration of the algorithm adds at most one element to $\cX$, we have by Claim~\ref{cl10-mnmxII} that $|\cX|=O\big(n\log \psi+\frac{n}{\epsilon^2}\big)$, where we set $r=n$, $\bar A:=\sum_{i=1}^rA_{i}\succ 0$, and use the set of matrices $\{p(q)p(q)^T:q\in\cS'\}$ for $A_1,\ldots,A_r$ in assumption (B-I), where $\cS'\subseteq\cS$ is a linearly independent subset of $\cS$. We bound  $\psi$ in the same way as in the proof of Claim~\ref{cl10-mnmxI}:
	\begin{align*}
	\psi&\le\frac{\max_{q\in\cS}Y'(0)\bullet p(q)p(q)^T+u(0)}{Y'(0)\bullet \frac1n\sum_{q\in\cS'}p(q)p(q)^T+u(0)}=\frac{\max_{q\in\cS}Y(0)\bullet qq^T+u(0)}{Y(0)\bullet \frac1n\sum_{q\in\cS'}qq^T+u(0)}\le \frac{n\cdot\max_{q\in\cS}\|q\|^2}{\lambda_{\min}(\sum_{q\in\cS'}qq^T)}\le n^2W^2\big(2^{2\cL'+1}n+1\big)=n^3W^{O(n^2)}.
	\end{align*} 
	It follows that $|\cX|=O(\frac{n^3}{\epsilon^2}\log (nW))$ (which is also a bound on the number of iterations of the algorithm).
	Moreover, by Remark~\ref{r2}, we have $\sum_{q\in\cX}\lambda_q\le 1+O(\epsilon)$. Thus scaling each $\lambda_q$ by $\sum_{q'\in\cX}\lambda_{q'}$ yields the sought convex combination satisfying the first inequality in~\raf{cvx-comb-}. 
	\bigskip
	
	Now consider the minimization problem. (In in this part of the proof, we do not require $\cS$ to be full-dimensional.) We can write \raf{cvx-}-\raf{cvx-dual-} in the form of \raf{packII}-\raf{coverII}, where the set of constraints $[m]$ corresponds to $\cS$, and where $A_q$, $C$ and $X$ are given by~\raf{eeeee0}. 
	By the reduction in Appendix~\ref{normal-II}, we can reduce~\raf{cvx-}-\raf{cvx-dual-} to normalized form without changing the value of the objective, but we need to show that each step of this reduction can be implemented in polynomial time. Consider assumption (C-II). Suppose that this assumption does not hold. Then there is an $x\in\RR^n$ such that $Qx=0$ and $q^Tx\ne 0$ for some $q\in\cS$, implying that $q\not\in\im(Q):=\{Qv~:~v\in\RR^n\}$. Conversely, if $q\not\in\im(Q)$, then (by Farkas' Lemma) there exists an $x\in\RR^n$ such that $Qx=0$ and $q^Tx\ne 0$. We conclude (by the same argument following assumption (C-II) in Appendix~\ref{normal-II}) that for $q\in\cS\setminus\im(Q)$,  the primal variable $\lambda_q=0$, and hence, we may replace $\cS$ by $\cS':=\cS\cap\im(Q)$ in \raf{cvx-dual-}. Let $\alpha Q=L^TD L$ be the LDL-decomposition of $\alpha Q$, and write $U=[U'~|~U'']:=L^{-1}$, where $U'$ is the submatrix of $U$ whose columns correspond to the columns of the submatrix $D'$ containing the {\it positive} diagonal entries of the diagonal matrix $D$.  Let $p(q):=(D')^{-1/2}(U')^Tq$, for $q\in\cS'$. Then~\raf{ecv0-} becomes equivalent to $\sum_{q\in\cS'}\lambda_q p(q)p(q)^T\preceq I$. 
	Next, we need to show that Assumption (B-II) can be made to hold in polynomial time. For our purposes, it is enough to show  a weaker version of this assumption, as we shall see below.
	We  begin by (implicitly) perturbing $p(q)p(q)^T$ into $\tilde A_q:=p(q)p(q)^T+\frac{\epsilon}nI$, for $q\in\cS'$. By the argument following Assumption~(B-II) in Appendix~\ref{normal-II}, $\frac1\beta\le z^*_{II}=1\le\frac n\beta$, where $\beta:=\min_{q\in\cS'}\|p(q)\|^2$, from which we obtain that $1\le \beta\le n$. Furthermore, by the same argument, the optimal value $\tilde z_{II}$ of the perturbed problem satisfies $1-2\epsilon\le\tilde z_{II}\le 1$.
	Then, in view of Remark~\ref{r3}, it is enough to show that in each iteration $t$ of Algorithm~\ref{log-cover-alg}, we can find efficiently a $q\in\cS'$ such that $\tilde A_q\bullet Y'+u\le 1+O(\eps_s)$ for given $Y'=Y'(t)\succeq 0$ and $u=u(t)\ge 0$ such that $\Tr(Y')+u\in (1-\eps_s,1)$ (by Claim~\ref{cl0--mnmxII}, where $X(t):=\left(\begin{array}{ll}Y'&0\\0&u\end{array}\right)$ in step~\ref{s3.-mnmxII} of the algorithm). To do this, let $\cL'$ be the total bit length needed to describe $Q$ and $\{v_1,\ldots,v_k\}$ be a basis of $\nll(Q):=\{x\in\RR^n:Qx=0\}$. Note that, for each $i\in[k]$, each nonzero component of $v_i$ is bounded in absolute value from below by $2^{-\cL'}$ (see. e.g., \cite[chapter 1]{GLS88}). Given $Y'\succeq 0$ and $u\ge 0$, we apply $\cA$ to $(Y,Q)$, where $Y:=U'(D')^{-\frac12}Y'(D')^{-\frac12}(U')^T+\gamma\sum_{i=1}^kv_iv_i^T$ and  $\gamma:=2^{2\cL'}\alpha Q\bullet Y+1=2^{2\cL'}\alpha Q\bullet U'(D')^{-\frac12}Y'(D')^{-\frac12}(U')^T+1$, to get a $q\in\cS$ such that $qq^T\bullet Y\le\alpha Q\bullet Y$. We claim that $q\in\cS'$. For this, it is enough to show that $q^Tv_i=0$, for all $i\in[k]$. Suppose $v_i^Tq\ne0$ for some $i\in[k]$. Then $|v_i^Tq|\ge 2^{-L}$, implying that 
	$$qq^T\bullet Y= qq^T\bullet U'(D')^{-\frac12}Y'(D')^{-\frac12}(U')^T +\gamma\sum_{i=1}^k(q^Tv_i)^2\ge (2^{2\cL'}\alpha Q\bullet Y+1)2^{-2\cL'}>\alpha Q\bullet Y,$$
	a contradiction. We conclude that $q\in\cS'$, and moreover that $p(q)p(q)^T\bullet Y'=qq^T\bullet Y\le \alpha Q\bullet Y=(L')^TD'L'\bullet Y=I\bullet Y'\leq1-u$.
	Then, $\tilde A_q\bullet Y'+u\le 1+\frac{\epsilon}nI\bullet Y'<1+\bareps$, as required. 
	To bound the number of iterations of the algorithm, we need to specify which $q'$ is used initially.  This is done as follows. We start the algorithm by setting $Y'=I$ and applying the integrality gap verifier $\cA$ to $(Y,Q)$, as above, to obtain a $q'\in\cS'$ such that 
	\begin{align}\label{eee5-}
	\|p(q')\|^2&=p(q')p(q')^T\bullet Y'=q'q'^T\bullet Y\le\alpha Q\bullet Y=\alpha Q\bullet U'(D')^{-1}(U')^T\nonumber\\
	&=L^T\left[\begin{array}{ll} D'&0\\0&0\end{array}\right]L\bullet U\left[\begin{array}{ll} (D')^{-1}&0\\0&0\end{array}\right]U^T=\left[\begin{array}{ll} D'&0\\0&0\end{array}\right]\bullet \left[\begin{array}{ll} (D')^{-1}&0\\0&0\end{array}\right]\le n.
	\end{align}
	Let $\cX\subseteq\cS$ be the set of vectors $q\in\cS$ such that $\lambda_q>0$ when the algorithm terminates. Since each iteration of the algorithm adds at most one element to $\cX$, we have by Claim~\ref{cl10-mnmxI} that $|\cX|=O\big(n\log \frac1\psi+\frac{n}{\epsilon^2}\big)$, where
	\begin{align*}
	\psi&=
	\frac{\lambda_{\min}(\widetilde A_{q(0)})}{\lambda_{\max}(\widetilde A_{q'})}
	\ge\frac{\epsilon/n}{n+\epsilon/n}\ge\frac{\epsilon}{2n^2}.
	\end{align*} 
	It follows that $|\cX|=O(n\log \frac{n}{\epsilon}+\frac{n}{\epsilon^2})$.  
	Moreover, by Remark~\ref{r3}, we have $\sum_{q\in\cX}\lambda_q\ge 1-O(\epsilon)$. Thus scaling each $\lambda_q$ by $\sum_{q'\in\cX}\lambda_{q'}$ yields the sought convex combination satisfying the second inequality in~\raf{cvx-comb-}.
\end{proof}

Note that, once we have a set $\cX$ as in Theorem~\ref{CV-SDP-approx}, its support can be reduced to $O(\frac{n^2}{\epsilon})$ using the sparsification techniques of \cite{BSS14,SHS16}. Applications of the Carr-Vempala type decomposition for SDPs in  robust {\it discrete} optimization can be found in \cite{EMO18}.

%\section{Extension to Convex/concave Functions}
%
%Let $f,h:\SS_+^n\to\SS_+^n$ be concave and convex matrix functions, respectively. Consider the following generalization of \raf{ncoverII} and \raf{npackII}: 
%
%{\centering \hspace*{-18pt}
%	\begin{minipage}[t]{.47\textwidth}
%		\begin{alignat}{3}
%		\label{gcover}
%		\tag{\GC} \quad& \displaystyle z_I^* = \min\quad \b1^Ty\\
%		\text{s.t.}\quad & \displaystyle f\left(\sum_{i=1}^my_iA_i\right)\succeq I\nonumber\\
%		\qquad &y\in\RR^m,~y\geq 0\nonumber
%		\end{alignat}
%	\end{minipage}
%	\,\,\, \rule[-14ex]{1pt}{14ex}
%	\begin{minipage}[t]{0.47\textwidth}
%		\begin{alignat}{3}
%		\label{gpack}
%		\tag{\GP} \quad& \displaystyle z_{II}^* = \min\quad \b1^Ty\\
%		\text{s.t.}\quad & \displaystyle h\left(\sum_{i=1}^my_iA_i\right)\preceq I\nonumber\\
%		\qquad &y\in\RR^m,~y\geq 0\nonumber.
%		\end{alignat}
%\end{minipage}}
%
%\medskip
%\noindent In this section we show how the algorithms for linear functions in Section~\ref{log-alg} can be extended to the concave/convex case. Without loss of generality, we consider \raf{gcover}. The algorithm is a slight modification of Algorithm~\ref{log-pack-alg}, where for any $i\in[m]$, we replace $A_i$ by $F(\b1_i)$, with $F(y(t)):=f\left(\sum_{i=1}^my_iA_i\right)$. Note that our assumption on $f(\cdot)$ implies that 

\appendix

\section{A Matrix MWU Algorithm for~\raf{packII}-\raf{coverII}}\label{MMWU-covering}
Given positive semidefinite matrices $A_1,\ldots,A_m \in \mathbb{S}_+^n$, we consider the dual packing-covering pair ~\raf{npackII}-\raf{ncoverII}. 
Here is a matrix MWU algorithm.

\begin{algorithm}[H]
	\SetAlgoLined
	$t \gets 0$; $y(0)\gets 0$; $X(0)\gets 0$; $M(0)\gets0$; $T \gets \epsilon^{-2}\ln n$\\	
	\While{$M(t) < T$}{\label{s-term-pII}
		$P(t) =  (1+\epsilon)^{\sum_{i=1}^my_i(t)A_i }$ /* Update the weight matrix by exponentiation */\\
		$i(t) \gets\argmin_{i} A_i\bullet X(t)$ \label{s1-pII}\\ 
		$\delta(t) \gets 1/\lambda_{\max}(A_{i(t)})$ \label{s-step-pII} /* Define the update step size */\\
		$X(t+1)\gets X(t)+\frac{\delta(t)P(t)}{I\bullet X(t)}$; $y(t+1) \gets y(t) + \delta(t)\b1_{i(t)}$ 	/* Update the primal-dual solution */\\
		$M(t+1) \gets \lambda_{\max}(\sum_iy_i(t)A_i)$ /* Compute the largest eigenvalue of LHS of dual */\\
		$t \gets t+1$
	}
	$L(t)\gets\min_{i}A_i\bullet X(t)$ \\
	Output $(\hat X,\hat y)=\left(\frac{X(t)}{L(t)},\frac{y(t)}{M(t)}\right)$
	\caption{Matrix MWU Algorithm for~\raf{packII}-\raf{coverII}}\label{MWU-cover-alg}
\end{algorithm}

%Define the potential function:
%$$
%\Phi(t):=\Tr(X(t))=I\bullet X(t).
%$$
%\newpage

\bigskip

\subsection{Analysis}
Let $F(t):=\sum_{i=1}^my_i(t)A_i$. 

\subsubsection{Number of Iterations}
\begin{claim}\label{cl1-pII}
	The algorithm terminates in at most $nT$ iterations. Note that by Claim~\ref{cl9-pII} below, $L(t_f)>0$.
\end{claim}

\begin{proof}
	Note that $\sum_{j=1}^n\lambda_j(F(t))=I\bullet F(t)$. Then
	\begin{align*}
	\sum_{j=1}^n\lambda_j(F(t+1))-\sum_{j=1}^n\lambda_j(F(t))&=I\bullet F(t+1)-I\bullet F(t)=\delta(t)I\bullet A_{i(t)}\\
	&=I\bullet \frac{A_{i(t)}}{\lambda_{\max}(A_{i(t)})}=\frac{\Tr(A_{i(t)})}{\lambda_{\max}(A_{i(t)})}=\frac{\sum_j\lambda_j(A_{i(t)})}{\lambda_{\max}(A_{i(t)})}\ge 1.
	\end{align*}
	It follows that $\sum_{j=1}^n\lambda_j(F(nT))\ge nT$ and thus 
	$$
	\lambda_{\max}(F(nT))\ge \frac{\sum_{j=1}^n\lambda_j(F(nT))}{n}\ge T.
	$$
	The claim follows by the termination condition in step~\ref{s-term-pII}.
\end{proof}

Let $t_f$ be the value of $t$ when the algorithm terminates.

\subsubsection{Primal Dual Feasibility and Approximate Optimality  }
\begin{claim}\label{cl2-pII}
	(Primal and dual feasibility:) \ \ \  $A_i\bullet \hat X\ge 1~\forall i\in[m]$, $\hat{X}\succeq 0$,  and $\sum_{i=1}^m\hat y_iA_i\preceq I$.
\end{claim}
\begin{proof}
	For any $i\in[m]$, we have 
	\begin{align*}
	A_i\bullet\hat X&=\frac{A_i\bullet X(t_f)}{L(t_f)}=\frac{A_i\bullet X(t_f)}{\min_iA_i\bullet X(t_f)}\ge 1.
	\end{align*}
	Also, $\hat X(t)=\frac{1}{L(t_f)}\sum_{t=0}^{t_f-1}\frac{\delta(t)X(t)}{I\bullet X(t)}\succeq 0$, since $X(t)\succeq 0$. Thus the primal is feasible. To see dual-feasibility, note that
	\begin{align*}
	\sum_{i=1}^m\hat y_iA_i=\frac{F(t_f)}{M(t_f)}=\frac{F(t_f)}{\lambda_{\max}(F(t_f))}.
	\end{align*}
	Thus, $\lambda_{\max}(F(t_f))=1$, implying that $\sum_{i=1}^m\hat y_iA_i\preceq I$. 
	
\end{proof}

\begin{claim}\label{cl3-pII}
	$L(t)\ge\sum_{t'=0}^{t-1}\frac{\delta(t')P(t')\bullet A_{i(t')}}{I\bullet P(t')}$.
	
	%	\noindent{\it Hint:}  Lower-bound $A_i\bullet X(t)$ using the definition of $i(t)$ in step~\ref{s1-pII}. 
\end{claim}	
\begin{proof}
	For any $i\in[m]$, we have for all $t'$
	\begin{align*}
	A_i\bullet X(t')&=A_i\bullet \frac{\delta(t')P(t')}{I\bullet P(t')}=\frac{\delta(t')A_i\bullet P(t')}{I\bullet P(t')}\ge \frac{\delta(t')A_{i(t')}\bullet P(t')}{I\bullet P(t')}\tag{by the definition of $i(t)$ in step~\ref{s1-pII}.}
	\end{align*}
	Summing the above inequality over all $t'<t$, we get the claim.
\end{proof}

\begin{claim}\label{cl4-pII}
	$$
	I\bullet P(t+1)\le I\bullet (1+\epsilon)^{F(t)}(1+\epsilon)^{\delta(t)A_{i(t)}}.
	$$
\end{claim}
\begin{proof}
	We will use the Golden–-Thompson inequality (see, e.g., \cite{T15}): for any two symmetric matrices $B$ and $C$:
	$$
	\Tr(e^{B+C})\le\Tr(e^Be^C).
	$$
	Now, 
	\begin{align*}
	I\bullet P(t+1)&=\Tr\left(e^{\ln(1+\epsilon)(F(t)+\delta(t)A_{i(t)})}\right)\le \Tr\left((1+\epsilon)^{F(t)}(1+\epsilon)^{\delta(t)A_{i(t)}}\right) \tag{by the Golden–-Thompson inequality}\\
	&=I\bullet (1+\epsilon)^{F(t)}(1+\epsilon)^{\delta(t)A_{i(t)}}.
	\end{align*}
\end{proof}

\begin{fact}\label{f1-pII}
	For $0\preceq B \preceq I$ and $\epsilon>0$, $$(1+\epsilon)^B\preceq I+\epsilon B.$$	
\end{fact}
\begin{proof}
	Let $B=U^T\Lambda U$ be the eigen decomposition of $B$, where $\Lambda=\diag(\lambda_1,\ldots,\lambda_n)$. Then 
	\begin{align}
	(1+\epsilon)^B-(I+\epsilon B)&=U^T\diag\left((1+\epsilon)^{\lambda_1},\ldots,(1+\epsilon)^{\lambda_n}\right)U-U^T\diag\left(1+\epsilon\lambda_1,\ldots,1+\epsilon\lambda_n\right)U\nonumber\\
	&=U^T\diag\left((1+\epsilon)^{\lambda_1}-(1+\epsilon\lambda_1),\ldots,(1+\epsilon)^{\lambda_n}-(1+\epsilon\lambda_n)\right)U.\label{eee1}
	\end{align}
	Using the inequality: $(1+\epsilon)^x \leq 1+\epsilon x$, valid for for $x \in [0,1]$ and $\epsilon>0$, we obtain that $(1+\epsilon)^{\lambda_j}-(1+\epsilon\lambda_j)\le 0$ and the claim follows from~\raf{eee1}. 
\end{proof}
\begin{claim}\label{cl5-pII}
	$$(1+\epsilon)^{\delta(t)A_{i(t)}}\preceq I+\epsilon \delta(t)A_{i(t)}.$$
\end{claim}
\begin{proof}
	The claim follows from Fact~\ref{f1-pII}, applied with $B:=\delta(t)A_{i(t)}$, which satisfies $0\preceq B \preceq I$ by the definition of $\delta(t)$ in step~\ref{s-step-pII} of the algorithm.
\end{proof}

\begin{fact}\label{f2-pII}
	For three symmetric matrices $B,C,D\in\RR^{n\times n}$ if $B\succeq 0$ and $C\preceq D$ then
	$$B\bullet C \le B\bullet D.$$
\end{fact}
\begin{proof}
	Immediate from $B\bullet(D-C)\ge 0$ which holds by the positive semidefiniteness of $B$ and $D-C$.
\end{proof}

\begin{claim}\label{cl6-pII}
	$$
	I\bullet P(t+1)\le I\bullet P(t)\left(1+\frac{\epsilon\delta(t)P(t)\bullet A_{i(t)}}{I\bullet P(t)}\right).
	$$
\end{claim}
\begin{proof}
	We conclude from Claims~\ref{cl4-pII} and \ref{cl5-pII}, and Fact~\ref{f2-pII} applied with $B:=(1+\epsilon)^{F(t)}$, $C:=(1+\epsilon)^{\delta(t)A_{i(t)}}$ and $D:=I+\epsilon\delta(t)A_{i(t)}$, that 
	\begin{align*}
	I\bullet P(t+1)&\le(1+\epsilon)^{F(t)}\bullet(1+\epsilon)^{\delta(t)A_{i(t)}}\le (1+\epsilon)^{F(t)}\bullet\left(I+\epsilon\delta(t)A_{i(t)}\right)\\
	&=I\bullet P(t)\left(1+\frac{\epsilon\delta(t)P(t)\bullet A_{i(t)}}{I\bullet P(t)}\right).
	\end{align*}
\end{proof}

\begin{claim}\label{cl7-pII}
	$I\bullet X(t)\le I\bullet P(0)e^{\epsilon\sum_{t'=0}^{t-1}\frac{\delta(t')P(t')\bullet A_{i(t')}}{I\bullet P(t')}}$.
\end{claim}
\begin{proof}
	Using the inequality $1+x\leq e^x$, valid for all $x\in\RR$, we get from Claim~\ref{cl7-pII}, 
	\begin{align}\label{eeee1}
	I\bullet P(t'+1)\le e^{\frac{\epsilon\delta(t')P(t')\bullet A_{i(t')}}{I\bullet P(t')}}.
	\end{align}
	Iterating \raf{eeee1} for $t'=0,1,\ldots,t-1$, we arrive at the claim.
\end{proof}
\begin{claim}\label{cl8-pII}
	$
	M(t) \ln(1+\epsilon)\le \ln n+ \epsilon\sum_{t'=0}^{t-1}\frac{\delta(t')P(t')\bullet A_{i(t')}}{I\bullet P(t')}.
	$
\end{claim}
\begin{proof}
	Taking logs in Claim~\ref{cl7-pII}, and using that $I\bullet P(0)=n$ and  
	$$
	I\bullet X(t)=\sum_{j=1}^n\lambda_j((1+\epsilon)^{F(t)})=\sum_{j=1}^n(1+\epsilon)^{\lambda_j(F(t))}\ge (1+\epsilon)^{\lambda_{\max}(F(t))},
	$$
	we get 
	\begin{align*}
	M(t)\ln(1+\epsilon)&=\lambda_{\max}(F(t))\ln(1+\epsilon)\le \ln n+\epsilon\sum_{t'=0}^{t-1}\frac{\delta(t')P(t')\bullet A_{i(t')}}{I\bullet P(t')}.
	\end{align*}
\end{proof}
\begin{claim}\label{cl9-pII}
	$\frac{L(t_f)}{M(t_f)}\ge \frac{\ln(1+\epsilon)}{\epsilon}-\epsilon\ge 1-1.5\epsilon$ for $\epsilon\in(0,0.5]$.
\end{claim}	
\begin{proof}
	Using Claims~\ref{cl3-pII} and~\ref{cl8-pII}, we obtain after rearranging terms
	\begin{align*}
	\frac{L(t_f)}{M(t_f)}\ge \frac{\ln(1+\epsilon)}{\epsilon}-\frac{\ln n}{\epsilon M(t_f)}\\
	&\ge \frac{\ln(1+\epsilon)}{\epsilon}-\frac{\ln n}{\epsilon T}\tag{by the termination condition}\\
	&=\frac{\ln(1+\epsilon)}{\epsilon}-\epsilon \tag{by the definition of $T$}\\
	&\ge 1-1.5\epsilon \tag{$\because \frac{\ln(1+\epsilon)}{\epsilon} -\epsilon \geq 1-1.5\epsilon$ for $\epsilon \in (0,0.5].$}
	\end{align*}
\end{proof}
\begin{claim}\label{cl10-pII}
	$I\bullet X(t)=\b1^Ty(t)=\sum_{t'=0}^{t-1}\delta(t')$. Thus, the following (approximate strong duality) holds
	$$
	(1-1.5\epsilon)I\bullet \hat X\le \b1^T\hat y.
	$$
\end{claim}
\begin{proof}
	The first claim follows by
	$$
	I\bullet \Delta X(t)=\frac{I\bullet \delta(t)P(t)}{I\bullet P(t)}=\delta(t)=\b1^Ty(t).
	$$
	From this we get from Claim~\ref{cl9-pII}
	$$
	\frac{\b1^T\hat y}{I\bullet \hat X}=\frac{\b1^Ty(t_f)}{M(t_f)}\bigg/\frac{I\bullet X(t_f)}{L(t_f)}=\frac{L(t_f)}{M(t_f)}\ge 1-1.5\epsilon,
	$$
	from which the second claim follows.
\end{proof}	

\subsubsection{Running Time per Iteration}

The most expensive step is the matrix exponentiation. It can be done (approximately) in time $O(n^3)$, by computing the eigenvalue decomposition of $F(t)$ (more efficient algorithms are available if $F(t)$ is sparse,  see, e.g. \cite{IPS11}). 
\begin{theorem}\label{thm:cover-MWU}
	For any $\epsilon>0$, Algorithm~\ref{MWU-cover-alg} outputs an  $O(\frac{n\log n}{\epsilon^2})$-sparse $O(\epsilon)$-optimal primal-dual pair in time $O( \frac{n^4\log n}{\epsilon^2}+\frac{n\cT\log n}{\epsilon^2})$, where $\cT$ is the time taken by a single call to the minimization oracle in step~\ref{s1-pII}. 
\end{theorem}

\section{Reduction to Normalized Form}\label{normal}
When $C=I=I_n$, the identity matrix in $\RR^{n\times n}$ and $b=\b1$, the vector of all ones in $\RR^m$, we say that the packing-covering SDPs~\raf{packI}-\raf{coverI} and~~\raf{packII}-\raf{coverII} are in {\it normalized} form.
We recall below how a general packing covering pair of SDPs can be reduced to normalized form (see e.g., \cite{JY11}). We denote by $I$ the identity matrix of appropriate dimension.

We first note that under assumption~(A), strong duality holds for both pairs~\raf{packI}-\raf{coverI} and  \raf{packII}-\raf{coverII}. Indeed, it is enough for this (see, e.g.,\cite{NN94}) to show the strict feasibility of the primal (the so-called the Slater's condition). For~\raf{packI} (resp.,~\raf{coverII}), a strict primal feasible solution is given by $X:=\delta I$, where $\delta:=\frac1{2\max_{i}I\bullet A_i}$ (resp., $\delta:=\frac2{\min_{i}I\bullet A_i}$). 

\subsection{Reduction to Normalized Form for \raf{packI}-\raf{coverI}} \label{normal-I}

Under assumption~(B-I), we may further assume that

\begin{itemize}
	\item[\bf (C-I)]    $C\succ 0$ and hence $C=I$.
\end{itemize}	
If this is not the case, we slightly perturb the matrix $C$ to make it positive definite without changing the objective value by much\footnote{such a reduction has been used, e.g., in \cite{BGCL15}}. Let $C=L^TD L$ be the LDL-decomposition of $C$ and $\bar I$ be the $0/1$-diagonal matrix with ones only in the entries corresponding to the {\it zero} diagonal entries of the diagonal matrix $D$.  For $\delta>0$, define $D(\delta) := D+\delta \bar I\succ 0$, $C(\delta):=L^TD(\delta)L=C+\delta L^T\bar I L$, and let $z^*_I(\delta)$ be the common optimum value of \raf{packI}-\raf{coverI} when $C$ is replaced by $C(\delta)$, and $X^*(\delta)$ and $y^*(\delta)$ be corresponding optimal primal and dual solutions, respectively. 

By assumption (B-I), for any feasible solution $X$ to \raf{packI}, we have $I\bullet X\le\tau$. Also, since $X=\frac1{\max_iA_i\bullet I} I$ is feasible for \raf{packI}, $z^*_I\ge \zeta:=\frac{C\bullet I}{\max_iA_i\bullet I}$. Thus, for any desired accuracy $\epsilon>0$, selecting $\delta=\frac{\epsilon \zeta}{\tau L\bullet L^T}$ gives
\begin{align*}
z^*_I&\le z^*_I(\delta)=C\bullet X^*(\delta)+\delta L^T\bar I L\bullet X^*(\delta)=C\bullet X^*(\delta)+\delta \bar I\bullet L X^*(\delta)L^T\le C\bullet X^*(\delta)+\delta I\bullet L X^*(\delta)L^T\\
&= C\bullet X^*(\delta)+\delta L^TL\bullet X^*(\delta)\le C\bullet X^*(\delta)+\delta \lambda_{\max}(L^TL) I\bullet X^*(\delta)\le C\bullet X^*(\delta)+\delta I\bullet L^TL \cdot I\bullet X^*(\delta)\\ &\le C\bullet X^*(\delta)+\delta \tau L\bullet L^T=  C\bullet X^*(\delta)+\epsilon\zeta\le C\bullet X^*(\delta)+\epsilon z^*_I\le (1+\epsilon)z_I^*.
\end{align*}
It follows that $X^*(\delta)$ is feasible solution to~\raf{packI} with objective value $C\bullet X^*(\delta)\ge (1-\epsilon)z_I^*$. It follows also that $y^*(\delta)$ is feasible for~\raf{coverI} (as $\sum_iy^*(\delta)A_i\succeq C(\delta)\succ C$) with objective value $z^*_I(\delta)\le(1+\epsilon)z_I^*$.

\medskip

Finally,  writing $U:=L^{-1}$, and replacing $X$ by $X':=D(\delta)^{\frac12}LXL^TD(\delta)^{\frac12}$, $A_i$ by $A_i':=D(\delta)^{-\frac12}U^TA_iUD(\delta)^{-\frac12}$ and $C(\delta)$ by $C'=I$, we obtain an equivalent version of the perturbed \raf{packI}-\raf{coverI} in normalized form. Given an optimal primal solution $X'$ for the normalized problem we get a feasible solution $X=UD(\delta)^{-\frac12}X'D(\delta)^{-\frac12}U^T$ to  the perturbed ~\raf{packI} with the same objective value. Similarly an optimal dual solution for the normalized problem is an optimal solution for  the perturbed ~\raf{coverI} as $\sum_iy_iA_i'\succeq I \Leftrightarrow \sum_iy_iA_i\succeq C(\delta)$.   Note that this reduction can be implemented in $O(n^3+n^\omega m)$ time. Moreover, given a maximization oracle Max$(\cdot)$ for~\raf{packI}-\raf{coverI}, we obtain a maximization oracle for the normalized problem as follows: given $Y\in\SS_+^n$, we return Max$(Y')$ with $Y':= UD(\delta)^{-\frac12}YD(\delta)^{-\frac12}U^T$. (For simplicity we ignore roundoff errors resulting from computing square roots, which can be dealt with using standard techniques)

\bigskip

\subsection{Reduction to Normalized Form for \raf{packII}-\raf{coverII}}\label{normal-II}

For a matrix $B\in\RR^{n\times n}$, define $\supp(B):=\{x\in\RR^n:~Bx\ne 0\}$.

We may assume that

\begin{itemize}
	\item[{\bf (C-II)}] $\supp(C)\supseteq\bigcup_{i}\supp(A_i) $.  
\end{itemize}
If this is not the case, that is, there is an $i\in[m]$ such that $\supp(A_i)\not\subseteq\supp(C)$ then $y_i=0$ for any feasible solution $y$ to~\raf{packII}. Indeed, the existence of an $x\in\RR^n$ such that $A_ix\ne0$ and $Cx=0$ implies that $y_ix^TA_ix\le y_ix^TA_ix+\sum_{j\ne i}y_jx^TA_jx\le x^TCx=0$, giving that $y_i=0$. Furthermore, the existence of such $x$ allows us to remove the $i$th inequality from \raf{coverII}; given an optimal solution $X$ to the reduced covering system, we obtain a feasible solution with the same objective value (and hence optimal) to the original system by setting $X=X'+\frac{xx^T}{A_i\bullet xx^T}$. 
Note that (by Farkas' Lemma \cite[Chapter 7]{S86}) we can check if (C-II) holds by solving $m$ linear systems of equations  $C\Gamma=A_i$, for $i=1,\ldots,m$. This can be done in $O(n^3+n^\omega m)$ time, where $\omega$ is the exponent of matrix multiplication, by  computing the LDL-decomposition of $C$.  

We may assume next that
\begin{itemize}
	\item[{\bf (D-II)}] $\supp(C)=\RR^n\setminus\{0\}$ and hence $C=I$.  
\end{itemize}
Suppose that (D-II) does not hold. Let $C=L^TD L$ be the LDL-decomposition of $C$ and write $U=[U'~|~U'']:=L^{-1}$, where $U'$ is the submatrix of $U$ whose columns correspond to the columns of the submatrix $D'$ containing the {\it positive} diagonal entries of the diagonal matrix $D$. Then $U^TC U=D$ implies that $(U'')^TCU''=0$, which in turn implies that $CU''=0$ (since $C\succeq0$). The latter condition gives by (C-II) that $A_iU''=0$ for all $i$, and consequently, 
\begin{align}\label{e2-1}
U^TA_iU=\left[ 
\begin{array}{l}
(U')^T\\ \hline 
(U'')^T
\end{array}
\right]
A_i
\left[ 
\begin{array}{l|l}
U' & U''
\end{array}
\right]=\left[\begin{array}{l|l}
(U')^TA_iU' & (U')^TA_iU''\\ \hline
(U'')^TA_iU'&(U'')^TA_iU''
\end{array}
\right]=\left[\begin{array}{l|l}
(U')^TA_iU' & 0 \\ \hline
0&0
\end{array}
\right].
\end{align}
It follows that
\begin{align}\label{e2-2}
\sum_iy_iA_i\preceq C&\Longleftrightarrow\sum_iy_iU^TA_iU\preceq U^TCU=D=\left[
\begin{array}{l|l}
D'&0\\ \hline
0& 0
\end{array}
\right]\nonumber\\
&\Longleftrightarrow \sum_iy_i(U')^TA_iU'\preceq D'\nonumber\\
&\Longleftrightarrow \sum_iy_i(D')^{-\frac12}(U')^TA_iU'(D')^{-\frac12}\preceq I.
\end{align}
Thus, replacing $A_i$ by $A_i':=(D')^{-\frac12}(U')^TA_iU'(D')^{-\frac12}$ and $C$ by $I$, we obtain an equivalent dual problem in normalized form whose optimal solution $y$ is optimal for~\raf{packII}. Also, a feasible primal solution $X'$ to the corresponding normalized primal problem can be transformed to a feasible solution  $X=U'(D')^{-\frac12}X'(D')^{-\frac12}(U')^T$ to~\raf{coverII} with the same objective value, as  $C\bullet X=C\bullet   U'(D')^{-\frac12}X'(D')^{-\frac12}(U')^T=(D')^{-\frac12}(U')^TCU'(D')^{-\frac12}\bullet X'=I\bullet X'$, and 
$
A_i\bullet X=A_i\bullet   U'(D')^{-\frac12}X'(D')^{-\frac12}(U')^T=(D')^{-\frac12}(U')^TA_iU'(D')^{-\frac12}\bullet X'=A_i'\bullet X'$. Conversely, write $L^T=[(L')^T|(L'')^T]$, where $L'$ is the submatrix of $L$ whose rows correspond to the rows of the submatrix $D'$, and note by definition that $U'L'+U''L''=I$. Then, given any feasible solution $X$ to ~\raf{coverII}, a feasible solution to the normalized primal problem with the same objective value is given by $X':=(D')^{\frac12} L'X(L')^T(D')^{\frac12}$ since
\begin{equation*}\label{e2-3}
A_i'\bullet X'=(L')^T(U')^TA_i U'L' \bullet X=(I-U''L'')^TA_i(I-U''L'')\bullet X=A_i\bullet X,
\end{equation*}
and similarly $I\bullet X'=I\bullet(D')^{\frac12}L' X(L')^T(D')^{\frac12}=(L')^TD'L'\bullet X=C\bullet X$.
This step takes $O(n^3+n^\omega m)$ time. Moreover, given a minimization oracle Min$(\cdot)$ for~\raf{packII}-\raf{coverII}, we obtain a minimization oracle for the normalized problem as follows: given $Y\in\SS_+^n$, we return Max$(Y')$ with $Y':= U'(D')^{-\frac12}Y(D')^{-\frac12}(U')^T$.

\medskip

We may next make the following further assumption on ~\raf{npackII}-\raf{ncoverII}:

\begin{itemize}
	\item[{\bf (B-II$'$)}] $\frac{\epsilon\beta }{2n}\le\lambda_{\min}(A_i)\le\lambda_{\max}(A_i)\le\frac{3n\beta}{\epsilon}$, for all $i\in[m]$, where $\beta:=\min_i\lambda_{\max}(A_i)$.
\end{itemize}
(The same argument shows (B-II).)
Indeed, let $J:=\{i\in[m]:~\lambda_{\max}'(A_i)\le\frac{n\beta'}{\epsilon}\}$, and for $i\in J$, let $\widetilde A_i:=A_i+\frac{\epsilon\beta'}{n}I$, where $\lambda_{\max}'(A_i)$ is a $\frac12$-approximation $\lambda'_{\max}(A_i)$ of $\lambda_{\max}(A_i)$ and $\beta':=\min_{i}\lambda_{\max}'(A_i)$.  Consider the following pair of packing-covering SDP's:

{\centering \hspace*{-18pt}
	\begin{minipage}[t]{.47\textwidth}
		\begin{alignat}{3}
		\label{ncoverII-}
		\tag{\nC-$\widetilde{II}$} \quad& \displaystyle \tilde z_{II} = \min\quad I\bullet X\\
		\text{s.t.}\quad & \displaystyle \widetilde A_i\bullet X\geq 1, \forall i\in J\nonumber\\
		\qquad &X\in\RR^{n\times n},~X\succeq  0\nonumber
		\end{alignat}
		
	\end{minipage}
	\,\,\, \rule[-16ex]{1pt}{16ex}
	\begin{minipage}[t]{0.47\textwidth}
		\begin{alignat}{3}
		\label{npackII-}
		\tag{\nP-$\widetilde{II}$} \quad& \displaystyle \tilde z_{II} = \max\quad \sum_{i\in J}y_i\\
		\text{s.t.}\quad & \displaystyle \sum_{i\in J}y_i\widetilde A_i\preceq I\nonumber\\
		\qquad &y\in\RR^m,~y\geq 0.\nonumber
		\end{alignat}
\end{minipage}}

\noindent Note that $\lambda_{\min}(\widetilde A_i)\ge\frac{\epsilon\beta'}{n}\ge\frac{\epsilon\beta}{2n}$, while for $i\in J$, $\lambda_{\max}(\widetilde A_i)=\lambda_{\max}(A_i)+\frac{\epsilon\beta'}{n}\le\frac{2n\beta}{\epsilon}+\frac{\epsilon\beta}{n}\le \frac{3n\beta}{\epsilon}$.

Let us also note that $\frac1\beta\le z^*_{II}\le\frac n\beta$  (see.e.g., \cite{JY11}), as $\frac1\beta I$ is feasible for~\raf{ncoverII}, and if $X^*$ is optimal for~\raf{ncoverII}, then for $i\in\argmin_{i'}\lambda_{\max}(A_{i'})$, we have $I\bullet X^*\ge \frac{A_i\bullet X^*}{\lambda_{\max}(A_{i})}=\frac{A_i\bullet X^*}\beta\ge\frac1\beta$.

 Let us note next that if $X$ is feasible for~\raf{ncoverII}, then it is also feasible for~\raf{ncoverII-}. Hence,  $\tilde z_{II}\le z_{II}^*$. 
On the other hand, suppose $\widetilde X$ is optimal for~\raf{ncoverII-}. Then, $X:=\frac1{1-\epsilon}\left(\widetilde X+\frac{\epsilon}{\beta' n}I\right)$ is feasible for  ~\raf{ncoverII}, since for $i\in J$,
\begin{align*}
A_i\bullet X&=\frac1{1-\epsilon}\left(\widetilde A_i-\frac{\epsilon\beta'}{n}I\right)\bullet\left(\widetilde X+\frac{\epsilon}{\beta' n}I\right)\\
&=\frac1{1-\epsilon}\left(\widetilde A_i\bullet  \widetilde X-\frac{\epsilon\beta'}{n}I\bullet\widetilde X+\frac{\epsilon}{\beta' n}I\bullet\widetilde A_i-\frac{\epsilon^2}{n^2}I\bullet I\right)\\
&\ge \frac1{1-\epsilon}\left(1-\epsilon+\frac{\epsilon}{n}-\frac{\epsilon^2}{n}\right)\tag{$\because I\bullet\widetilde X=\tilde z_{II}\le z_{II}^*\le\frac n\beta\le \frac{n}{\beta'}$ and $I\bullet\widetilde A_i\ge I\bullet A_i\ge\beta\ge\beta'$}> 1,
\end{align*}
while for $i\not\in J$, we have  $A_i\bullet  X\ge A_i\bullet \frac{\epsilon}{\beta' n}I\ge\frac{\epsilon}{\beta' n}\lambda_{\max}(A_i)\ge\frac{\epsilon}{\beta' n}\lambda'_{\max}(A_i)>1.$ 
Moreover, $I\bullet X=\frac1{1-\epsilon}\left(I\bullet \widetilde X+\frac{\epsilon}{\beta'}\right)\le\frac1{1-\epsilon}\left(I\bullet \widetilde X+\frac{2\epsilon}{\beta}\right)\le\frac1{1-\epsilon}(\tilde z_{II}+2\epsilon z^*_{II})\le\frac{1+2\epsilon}{1-\epsilon}z^*_{II}$. 
Obviously, given a feasible solution $\tilde y$ to \raf{npackII-}, it can be extended to a feasible solution $y$ to \raf{npackII}, with the same objective value, by setting $y_i=\tilde y_i$ for $i\in J$ and $y_i=0$ for $i\not\in J$.   
To implement this step of the reduction, we can use Lanczos' algorithm to compute $\lambda_{\max}'(A_i)$, for $i\in[m]$. The running time is $O(mn^2)$.

\bigskip

\noindent{\bf Acknowledgment.}~We thank Waleed Najy for his help in the proofs in Section~\ref{pack-analysis}, %(yet declining to be a co-author), 
and for many helpful discussions.

%\bibliography{convex2journal}
%\bibliographystyle{plainurl}% the recommended bibstyle

\end{document}